\newcommand{\propose}[1]{{\color{black}#1}}

\newcommand\blfootnote[1]{%
  \begingroup
  \renewcommand\thefootnote{}\footnote{#1}%
  \addtocounter{footnote}{-1}%
  \endgroup
}

\documentclass[pldi-cameraready]{sigplanconf}

\usepackage{color}
\usepackage[usenames,dvipsnames]{xcolor}

\usepackage[font=small,skip=8pt]{caption}
\usepackage{pgfplots}
\usepackage{tikz}
\usetikzlibrary{plotmarks}
\usetikzlibrary{shapes.symbols}
\usetikzlibrary{calc}
\usetikzlibrary{patterns}
\usetikzlibrary{external}
\usepackage[amssymb]{SIunits}            %
\usepackage{courier}            %
\usepackage[scaled]{helvet} %
\usepackage{url}                  %
\usepackage{listings}          %
\usepackage[colorlinks=true,allcolors=blue,breaklinks,draft]{hyperref}   %
\usepackage{bbm}
\usepackage{amsmath}
\usepackage{amsthm}
\usepackage{amssymb}
\usepackage{graphicx}
\usepackage{paralist}
\usepackage{caption}
\usepackage{enumitem}      %
\usepackage{stmaryrd}
\usepackage{mathpartir}
\usepackage{balance}
\usepackage{titlesec}
\usepackage{lipsum}

\usepackage[caption=false]{subfig}
\usepackage{algorithm}
\usepackage[noend]{algpseudocode}

\usepackage{url}
\usepackage{multirow}
\usepackage{pgfplotstable}
\usepackage{booktabs}
\usepackage{dblfloatfix}

\usepackage{plstx}
\usepackage[makeroom]{cancel}
\usepackage{marginnote}
\usepackage{soul}

\usepackage{graphicx}
\usepackage{bbding}

\newcommand*\xmnote[3][0pt]{}

\makeatletter
\let\oldmarginnote\marginnote
\renewcommand*{\marginnote}[1]{%
   \begingroup%
   \ifodd\value{page}
     \if@firstcolumn\reversemarginpar\fi
   \else
     \if@firstcolumn\else\reversemarginpar\fi
   \fi
   \oldmarginnote{#1}%
   \endgroup%
}
\makeatother

\newcommand{\mnote}[1]{}

\usepackage[firstpage]{draftwatermark}
\SetWatermarkText{\hspace*{8in}\raisebox{6.65in}{\includegraphics[scale=0.1]{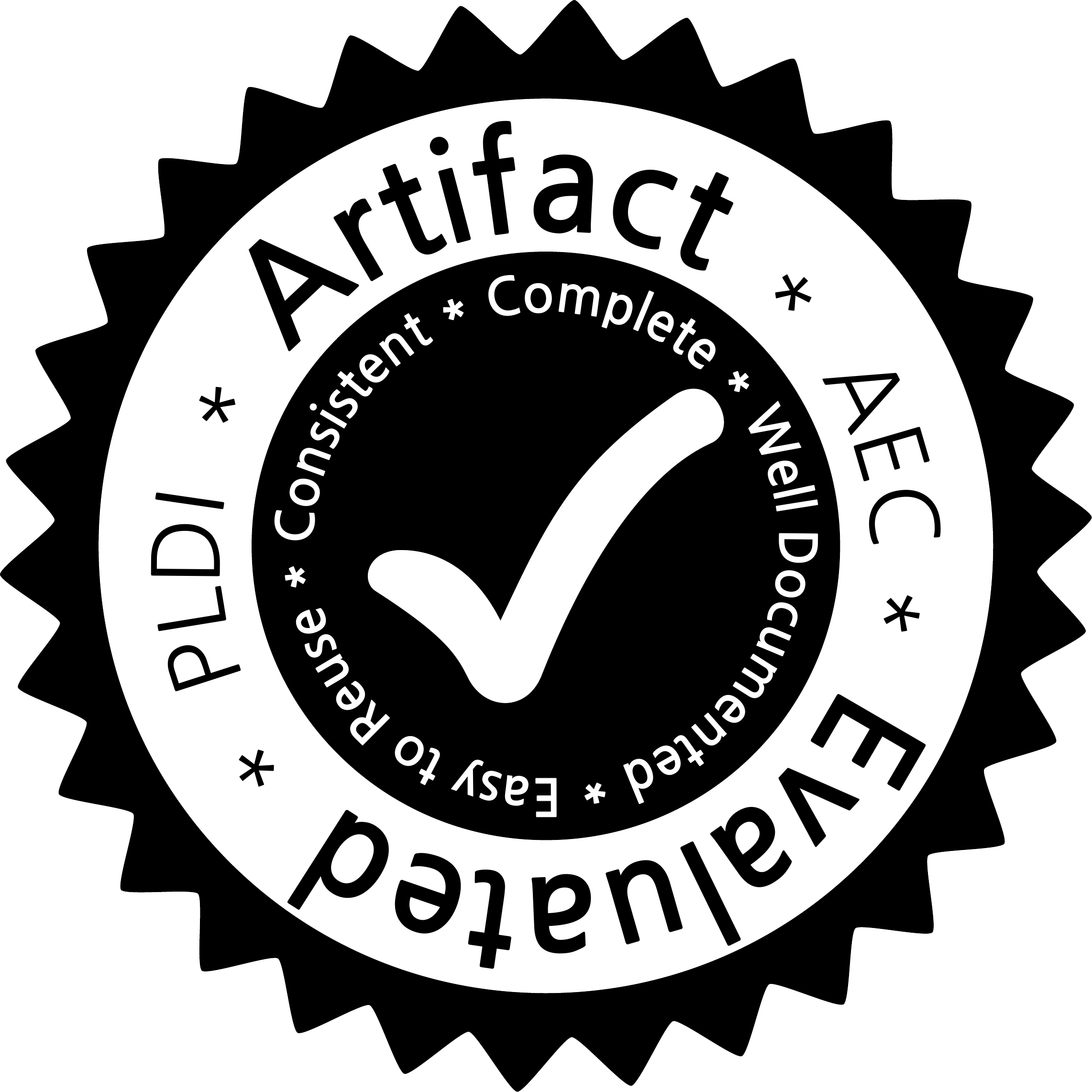}}}
\SetWatermarkAngle{0}

\newcommand{\comment}[1]{}
\newcommand{\port}{\ensuremath{\mathit{pt}}}

\newcommand{\switch}{\ensuremath{\mathit{sw}}}

\newcommand{\lpt}{\ensuremath{\mathit{lp}}}

\newcommand{\ntr}{\ensuremath{\mathit{ntr}}}

\newcommand{\fin}{\ensuremath{\mathit{fin}}}

\newcommand{\eid}{\ensuremath{\mathit{eid}}}

\newcommand{\coloneq}{\ensuremath{\mathord{::=}}}

\newcommand{\pt}{\ensuremath{\mathit{pt}}}

\newcommand{\host}{\ensuremath{\mathit{h}}}

\newcommand{\pkts}{\ensuremath{\mathit{pkts}}}
\newcommand{\pkt}{\ensuremath{\mathit{pkt}}}

\newcommand{\steps}[1]{\ensuremath{\xrightarrow{#1}}}
\newcommand{\listplus}{\ensuremath{\mathord{@}}}

\newcommand{\fld}{\ensuremath{\mathit{f}}}

\newcommand{\PTraces}{\ensuremath{\mathit{Traces}}}

\newcommand{\set}[2]{\ensuremath{#1 \mathord{:=} #2}}

\newcommand{\andalso}{\quad\;}

\newcommand{\cons}{\ensuremath{\mathord{::}}}

\makeatletter
\newcommand\incircbin
{%
  \mathpalette\@incircbin
}
\newcommand\@incircbin[2]
{%
  \mathbin%
  {%
    \ooalign{\hidewidth$#1#2$\hidewidth\crcr$#1\bigcirc$}%
  }%
}

\newcommand{\oeq}{\incircbin{=}}
\newcommand{\oneq}{\incircbin{\not=}}
\makeatother

\newcommand{\llrrparen}[1]{\llparenthesis #1 \rrparenthesis}

\definecolor{lightred}{RGB}{255,128,128}
\definecolor{lightblue}{RGB}{128,128,255}
\definecolor{lightgreen}{RGB}{128,255,128}
\definecolor{lightorange}{RGB}{255,204,128}
\definecolor{lightgray}{rgb}{0.98,0.98,0.98}

\definecolor{darkorange}{RGB}{255,132,0}
\definecolor{darkgreen}{RGB}{0,102,0}

\newcommand{\budget}[1]{}

\newtheorem{theorem}{Theorem}
\newtheorem{lemma}{Lemma}

\newtheorem{definition}{Definition}

\newcommand{\kw}[1]{\text{\tt\bf #1}}

\DeclareFontFamily{U}{mathb}{}
\DeclareFontShape{U}{mathb}{m}{n}{
  <-5.5> mathb5
  <5.5-6.5> mathb6
  <6.5-7.5> mathb7
  <7.5-8.5> mathb8
  <8.5-9.5> mathb9
  <9.5-11.5> mathb10
  <11.5-> mathbb12
}{}
\DeclareSymbolFont{mathb}{U}{mathb}{m}{n}
\DeclareMathSymbol{\sqcdot}{\mathbin}{mathb}{"0D}%

\lstdefinelanguage{scala}{
  morekeywords={abstract,case,catch,class,def,%
    do,else,extends,false,final,finally,%
    for,if,implicit,import,match,mixin,%
    new,null,object,override,package,%
    private,protected,requires,return,sealed,%
    super,this,throw,trait,true,try,%
    type,val,var,while,with,yield},
  otherkeywords={=>,<-,<\%,<:,>:,\#,@},
  sensitive=true,
  morecomment=[l]{//},
  morecomment=[n]{/*}{*/},
  morestring=[b]",
  morestring=[b]',
  morestring=[b]"""
}

\newcommand{\spacehack}[1]{}

\def\implTraceFig{
\begin{figure*}[t]
\footnotesize
\renewcommand{\dots}{\ensuremath{..}}
\fbox{~\begin{minipage}{0.975\textwidth}
\(
\begin{array}{@{~}l@{~~}|@{~~}l@{~~}|@{~~}l@{~}}
\begin{array}{llcl}
\textit{Switch ID} & n  & \in & \mathbb{N}\\
\textit{Port ID} & m  & \in & \mathbb{N}\\
\textit{Host ID} & \host  & \in & \mathbb{N}\\
\textit{Location} & l & \coloneq & n:m\\
\end{array} & \begin{array}{llcl}
\textit{Packet} & \pkt & \coloneq & \{ \fld_1; \cdots; \fld_k; C; digest \}\\
\textit{Located Packet} & \lpt & \coloneq & (\pkt,l)\\
\textit{Queue Map} & qm & \coloneq & \{ n \mapsto \pkts, \cdots \} \\
\textit{Link} & lk & \coloneq & (l, l) \\
\textit{Links} & L & \coloneq & \{lk,\cdots\}\\
\textit{Event} & e & \coloneq & (\varphi, l) \\
\textit{Event-set} & E & \coloneq & \{e, \cdots\}\\
\end{array} & \begin{array}{llcl}
\textit{Configuration} & C & \coloneq & \{(\lpt, \lpt), \cdots\} \\
\textit{Enabling Rel.} & \vdash & \coloneq & \{(E,e), \cdots \} \\
\textit{Consist. Pred.} & con & \coloneq & \{E, \cdots\} \\
\textit{Config. Map} & g & \coloneq & \{ E \mapsto C, \cdots \} \\
\textit{Switch} & sw & \coloneq & (n,qm, E, qm) \\
\textit{Queue, Control.} & Q,R & \coloneq & E \\
\textit{Switches} & S & \coloneq & \{sw, \cdots\}\\
\end{array}
\end{array}
\)

\smallskip
\hrule
\vspace{-0pt}
\begin{mathpar}
\hspace*{-1.5pc}
\inferrule*[Right=In]{ 
  (h,n{:}m) \in L  \andalso S = S' {\cup} \{(n,qm[m {\mapsto} \pkts],E,qm_2)\}
}{ 
  (Q,R,S)
  \steps{} 
  (Q,R,S' {\cup} \{(n,qm[m {\mapsto} \pkts\listplus[\pkt[C {\gets} {g(E)]}]],E,qm_2)\})
}
\and
\hspace{-3mm}\inferrule*[Right=Out]{ 
  (n{:}m,h) \in L \andalso S = S' {\cup} \{(n,q_1,E,qm[m {\mapsto} \pkt\cons\pkts])\}
}{ 
  (Q,R,S)
  \steps{} 
  (Q,R,S' {\cup} \{(n,q_1,E,qm[m {\mapsto} \pkts])\})
} \vspace{-0pt} \\
\hspace*{-5mm}
\inferrule*[Right=Switch]{ 
  E'=\{e:(E\cup\pkt.digest)\vdash e\land con(E \cup \pkt.digest \cup \{e\}) \land(\pkt,n{:}m)\models e\}
  \andalso \\ \{\lpt : \pkt.C((\pkt,n{:}m),\lpt)\}=\{(\pkt_1,n{:}m_1),\cdots\}
  \andalso S=S' \cup \{(n,qm[m \mapsto\pkt\cons\pkts],E,qm_2[m_1 \mapsto\pkts_1,\cdots])\}
}{ 
  (Q,R,S)
  \steps{} 
  (Q \cup E',R, S' \cup \{(n,qm[m \mapsto\pkts],  E  \cup  E'  \cup  \pkt.digest, \\ qm_2[m_1 \mapsto\pkts_1\listplus[\pkt_1[digest \gets \pkt_1.digest  \cup  E  \cup  E']],\cdots])\})
}
\end{mathpar}

\vspace{-10pt}

\begin{mathpar}
\inferrule*[Right=Link]{ 
  (n_1{:}m_1,n_2{:}m_2) \in L  \andalso S = S' \cup \{(n_1,qm_1,E_1,qm_2[m_1 \mapsto \pkt\cons\pkts]),(n_2,qm_3[m_2 \mapsto \pkts'],E_2,qm_4)\}
}{ 
  (Q,R,S)
  \steps{} 
  (Q,R,S' \cup \{(n_1,qm_1,E_1,qm_2[m_1 \mapsto \pkts]),(n_2,qm_3[m_2 \mapsto \pkts'\listplus[\pkt]],E_2,qm_4)\})
} \vspace{-0pt} \\
\and
\hspace*{-1cm}
\inferrule*[Right=CtrlRecv]{ 
  Q = Q' \cup \{e\}
}{ 
  (Q,R,S)
  \steps{} 
  (Q',R \cup \{e\}),S)
}
\and
\inferrule*[Right=CtrlSend]{ 
  R = R' \cup \{e\} \andalso
  S = S' \cup \{(n,qm,E,qm_2)\}
}{ 
  (Q,R,S)
  \steps{} 
  (Q,R,S' \cup \{(n,qm,E \cup \{e\},qm_2)\})
}
\end{mathpar}

\end{minipage}~}
\caption{Implemented program semantics.}
\label{fig:actual_model}
\spacehack{-.75em}
\end{figure*}%
}

    \renewcommand{\topfraction}{0.9}	%
    \renewcommand{\floatpagefraction}{0.7}	%

\begin{document}

\setlength{\pdfpageheight}{\paperheight}
\setlength{\pdfpagewidth}{\paperwidth}

\toappear{}

\makeatletter
\def\@ivtitleauthors#1#2#3#4{%
  \if \@andp{\@emptyargp{#2}}{\@emptyargp{#3}}%
    \noindent \@setauthor{40pc}{#1}{\@false}\par
  \else\if \@emptyargp{#3}%
    \noindent \@setauthor{17pc}{#1}{\@false}\hspace{3pc}%
              \@setauthor{17pc}{#2}{\@false}\par
  \else\if \@emptyargp{#4}%
    \noindent \@setauthor{17pc}{#1}{\@false}\hspace{3pc}%
              \@setauthor{17pc}{#3}{\@false}\par
  \else
    \noindent \@setauthor{9.3333pc}{#1}{\@false}\hspace{1.5pc}%
              \@setauthor{9.3333pc}{#2}{\@false}\hspace{1.5pc}%
              \@setauthor{9.3333pc}{#3}{\@false}\hspace{1.5pc}%
              \@setauthor{9.3333pc}{#4}{\@true}\par
    \relax
  \fi\fi\fi
  \vspace{10pt}} %

\def \@maketitle {%
  \begin{center}
  \@settitlebanner
  \let \thanks = \titlenote
  {\leftskip = 0pt plus 0.25\linewidth
   \rightskip = 0pt plus 0.25 \linewidth
   \parfillskip = 0pt
   \spaceskip = .7em
   \noindent \LARGE \bfseries \@titletext \par}
  \vskip 6pt %
  \noindent \Large \@subtitletext \par
  \vskip 22pt %
  \ifcase \@authorcount
    \@latex@error{No authors were specified for this paper}{}\or
    \@titleauthors{i}{}{}\or
    \@titleauthors{i}{ii}{}\or
    \@titleauthors{i}{ii}{iii}\or
    \@ivtitleauthors{i}{ii}{iii}{iv}\or
    \@titleauthors{i}{ii}{iii}\@titleauthors{iv}{v}{}\or
    \@titleauthors{i}{ii}{iii}\@titleauthors{iv}{v}{vi}\or
    \@titleauthors{i}{ii}{iii}\@titleauthors{iv}{v}{vi}%
                  \@titleauthors{vii}{}{}\or
    \@titleauthors{i}{ii}{iii}\@titleauthors{iv}{v}{vi}%
                  \@titleauthors{vii}{viii}{}\or
    \@titleauthors{i}{ii}{iii}\@titleauthors{iv}{v}{vi}%
                  \@titleauthors{vii}{viii}{ix}\or
    \@titleauthors{i}{ii}{iii}\@titleauthors{iv}{v}{vi}%
                  \@titleauthors{vii}{viii}{ix}\@titleauthors{x}{}{}\or
    \@titleauthors{i}{ii}{iii}\@titleauthors{iv}{v}{vi}%
                  \@titleauthors{vii}{viii}{ix}\@titleauthors{x}{xi}{}\or
    \@titleauthors{i}{ii}{iii}\@titleauthors{iv}{v}{vi}%
                  \@titleauthors{vii}{viii}{ix}\@titleauthors{x}{xi}{xii}%
  \else
    \@latex@error{Cannot handle more than 12 authors}{}%
  \fi
  \vspace{1.75pc} %
  \end{center}}
\makeatother

\def\titletext{Event-Driven Network Programming}
\title{\titletext}
\preprintfooter{\titletext}

\authorinfo{Jedidiah McClurg}
           {CU Boulder, USA}
           {\small jedidiah.mcclurg@colorado.edu}
\authorinfo{Hossein Hojjat}
           {Cornell University, USA}
           {\small hojjat@cornell.edu}
\authorinfo{Nate Foster}
           {Cornell University, USA}
           {\small jnfoster@cs.cornell.edu}
\authorinfo{Pavol {\v C}ern\'y}
           {CU Boulder, USA}
           {\small pavol.cerny@colorado.edu}

\maketitle

\def\naive/{na\"{\i}ve}

\begin{abstract}
Software-defined networking (SDN) programs must simultaneously
describe static forwarding behavior and dynamic updates in response to
events. Event-driven updates are critical to get right, but
difficult to implement correctly due to the high degree of concurrency
in networks. Existing SDN platforms offer weak guarantees that can 
break application invariants, leading to problems such as dropped
packets, degraded performance, security violations, etc. This paper
introduces {\em event-driven consistent updates} that are
guaranteed to preserve well-defined behaviors when transitioning
between configurations in response to events. We propose {\em network
  event structures} (NESs) to model constraints on updates, such as
which events can be enabled simultaneously and causal dependencies
between events. We define an extension of the NetKAT language with
mutable state, give semantics to stateful programs using NESs, and
discuss provably-correct strategies for implementing NESs in SDNs.
Finally, we evaluate our approach empirically,
demonstrating that it gives well-defined consistency guarantees while
avoiding expensive synchronization and packet buffering.
\end{abstract}



\section{Introduction}
\label{sec:introduction}

Software-defined networking (SDN) allows network behavior to be
specified using logically-centralized programs that execute on
general-purpose machines. These programs react to events such as
topology changes, traffic statistics, receipt of packets, %
etc. by modifying sets of forwarding rules installed on
switches. SDN programs can implement a wide range of advanced network
functionality including fine-grained access
control~\cite{ethane-sigcomm07}, network
virtualization~\cite{koponen2014network}, traffic
engineering~\cite{b4-sigcomm13,swan-sigcomm13}, and many others.

Although the basic SDN model is simple, building sophisticated
applications is challenging in practice. Programmers must keep track
of numerous low-level details such as encoding configurations into
prioritized forwarding rules, processing concurrent events, managing
asynchronous events, dealing with unexpected failures, etc. To address
these challenges, a number of domain-specific network programming
languages have been
proposed~\cite{anderson2014netkat,nelson2014tierless,foster2011frenetic,voellmy2013maple,soule2014merlin,kang2013optimizing,moshref2013scalable,kim2015kinetic}.
The details of these languages vary, but they all offer higher-level
abstractions for specifying behavior (e.g., using mathematical
functions, boolean predicates, relational operators, etc.), and rely
on a compiler and run-time system to generate and manage the
underlying network state.

Unfortunately, the languages that have been proposed so far lack
critical features that are needed to implement dynamic, event-driven
applications.  Static languages such as
NetKAT~\cite{anderson2014netkat} offer rich constructs for describing
network configurations, but lack features for responding to events and
maintaining internal state. Instead, programmers must write a stateful
program in a general-purpose language that generates a stream of
NetKAT programs. Dynamic languages such as FlowLog and
Kinetic~\cite{nelson2014tierless,kim2015kinetic} offer stateful
programming models, but they do not specify how the network behaves
while it is being reconfigured in response to state
changes. Abstractions such as consistent updates provide strong
guarantees during periods of
reconfiguration~\cite{reitblatt2012abstractions,mcclurg2015efficient},
but current realizations are limited to properties involving a single
packet (or set of related packets, such as a unidirectional flow). To
implement {\em correct} dynamic SDN applications today, the most effective option is
often to use low-level APIs, forgoing the benefits of
higher-level languages entirely.

\paragraph*{Example: Stateful Firewall.}
To illustrate the challenges that arise when
implementing dynamic applications, consider a topology where an
internal host $H_1$ is connected to switch $s_1$, an external host
$H_4$ is connected to a switch $s_4$, and switches $s_1$ and $s_4$ are
connected to each other (see Figure~\ref{fig:examples-topo-2}).
Suppose we wish to implement a stateful firewall: at all times,
host $H_1$ is allowed to send packets to
host $H_4$, but $H_4$ should only be allowed to send packets to $H_1$
if $H_1$ previously initiated a connection. Implementing even this
simple application turns out to be difficult, because it
involves coordinating behavior across multiple devices and packets.
\propose{%
The basic idea is that upon receiving a packet from $H_1$
at $s_4$, the program will need to issue a command to install a forwarding
rule on $s_4$ allowing traffic to flow from $H_4$ back to $H_1$.
  There are two
\xmnote{\FiveStar}{Q6}%
  straightforward (but incorrect) implementation strategies on current SDN controllers.

\begin{figure}[t]
\centering
\centerline{\includegraphics[trim = 0.0in 0.0in 0.0in 0.0in, clip,height=0.85in]{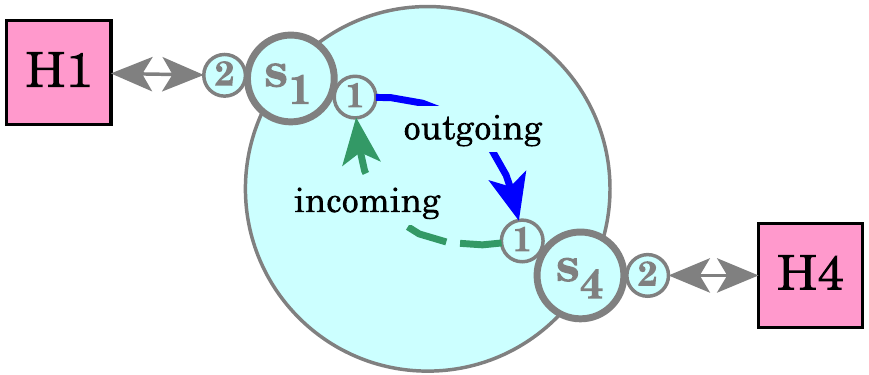}}
\caption{Topology for simple Stateful Firewall.}
\label{fig:examples-topo-2}
\end{figure}

  \begin{compactenum}
  \item The outgoing request from $H_1$ is diverted to the controller,
    which sets up flow tables for the incoming path and also forwards
    the packet(s) to $H_4$.  Reconfiguring flow tables takes time, so
    $H_4$'s response will likely be processed by the default drop
    rule. Even worse, if the response is the SYN-ACK in a TCP
    handshake, normal retransmission mechanisms will not help---the
    client will have to wait for a timeout and initiate another TCP
    connection. In practice, this greatly increases the latency of
    setting up a connection, and potentially wreaks havoc on
    application performance.

  \item The outgoing request is buffered at the controller, which sets
    up the flow tables for the incoming path but waits until the rules
    are installed before forwarding the packet(s).  This avoids the
    problem in (1), but places extra load on the controller and also
    implements the firewall incorrectly, since incoming traffic is
    allowed before the outgoing request is delivered.
      Leaving the network unprotected (even briefly)
      can be exploited by a malicious attacker.
  \end{compactenum}
  }

\noindent
\propose{Thus, while it is tempting to think that reliability
  mechanisms built into protocols such as TCP already prevent (or at
  least reduce) these types of errors, this is not the case.  While it
  is true that some applications can tolerate long latencies, dropped
  packets, and weak consistency, problems with updates {\em do} lead
  to serious problems in practice. As another example, consider an
  intrusion detection system that monitors suspicious traffic---%
inadvertently dropping or allowing even a few packets due to a
  reconfiguration would weaken the protection it provides.}
The root of these problems is that {\em existing SDN frameworks do not
  provide strong guarantees during periods of transition
  between configurations in response to events}.  An eventual
guarantee is not strong enough to implement the stateful firewall
correctly, and even a consistent
update~\cite{reitblatt2012abstractions} would not suffice, since consistent updates only
dictate what must happen to individual packets.

\propose{
\paragraph*{Existing Approaches.}

Experienced network 
\xmnote{\FiveStar}{Q1}%
operators may be able to use existing
 tools/methods to correctly implement event-driven configuration
 changes.
  However, as seen above, this requires
  thinking carefully about the potential interleavings of events and updates,
  delegating atomic operations to the controller (incurring a
  performance hit), etc.

As mentioned, %
\xmnote{\FiveStar}{Q3}%
there are stateful programming systems that attempt to make
this process easier for the programmer, but
update strategies in  
these systems
  either
  offer no consistency guarantees during dynamic updates,
  rely on expensive processing via the controller,
  and/or require the programmer to craft an update protocol by hand.
  In this paper,
  \xmnote{\FiveStar}{Q2}%
  we group these approaches together, using the term
  {\em uncoordinated update} to describe their lack of support for coordinating
  local updates in a way that ensures global consistency.
}

\paragraph*{Event-Driven Consistent Update.}
We propose a new semantic correctness condition with
clear guarantees about updates triggered by events. This
enables specification of how the network should behave during
updates, and enables %
precise formal reasoning
about stateful network programs.

An {\em event-driven consistent update} is denoted as a triple 
$C_i \xrightarrow{e} C_f$,
where $C_i$ and $C_f$ are the initial and final configurations respectively,
and $e$ is an 
event. Intuitively, these configurations describe the forwarding
behaviors of the network before/after the update, while the event
describes a phenomenon, such as the receipt of a packet at a particular switch,
that triggers the update itself.
Semantically, an event-triggered consistent update ensures that for each packet:
\begin{compactenum}
\item {\em the packet is forwarded consistently,} i.e.
   it must be processed entirely by a {\em single configuration}
  $C_i$ or $C_f$, and
\item {\em the update does not happen too early,} meaning that if every %
  switch traversed by the packet has {\em not} heard about the event, then the
  packet must be processed by $C_i$, and
\item {\em the update does not happen too late,} meaning that if every %
  switch traversed by the packet {\em has} heard about the event, then the packet must
  be processed by $C_f$.
\end{compactenum}
The first criterion requires that updates are consistent, which is
analogous to a condition proposed previously by Reitblatt et
al.~\cite{reitblatt2012abstractions}. However, a consistent update
alone would not provide the  necessary guarantees for the stateful
firewall example, as it applies only to a single packet, and not to
multiple packets in a bidirectional flow. The last two criteria relate
the packet-processing behavior on each switch to the events it has
``heard about.'' Note that these criteria leave substantial
flexibility for implementations: packets that do not satisfy the
second or third condition can be processed by either the preceding or following
configuration. It remains to define what it means for a switch $s$ to
have ``heard about'' an event $e$ that occurred at switch $t$
(assuming $s \neq t$). We use
a causal model and say that $s$ hears about $e$ when a packet, which
was processed by $t$ after $e$ occurred, is received at $s$. This
can be formalized using a ``happens-before''
relation. 

Returning to the stateful firewall, it is not hard to see that the
guarantees offered by event-driven consistent updates are
sufficient to ensure correctness of the overall
application. Consider an update $C_i \xrightarrow{e} C_f$. 
In $C_i$, $H_1$ can send packets to $H_4$, but not vice-versa. In
$C_f$, additionally $H_4$ can send packets to $H_1$. The event $e$ is
the arrival at $s_4$ of a packet from $H_1$ to $H_4$.
Before $e$ occurs, can $H_4$ send a packet to $H_1$, as is possible in $C_f$? No,
since {\em none} of the switches along the necessary path have
heard about the event.
Now, imagine that the
event $e$ occurs, and $H_4$ wants to send a packet to $H_1$
afterwards. Can $s_4$ drop the new packet, as it would have done in the
initial configuration $C_i$? No, because the {\em only} switch the packet would
traverse is $s_4$, and $s_4$ has heard about the event, meaning
that the only possible correct implementation should process this new packet in
$C_f$.

\paragraph*{Event-Driven Transition Systems.}
To specify event-driven network programs, we use labeled
transition systems called {\em event-driven transition systems} (ETSs). In
an ETS, each node is annotated with a network configuration and each
edge is annotated with an event. For example, the stateful
firewall application would be described as a two-state ETS,
one state representing the initial configuration before
$H_1$ has sent a packet to $H_4$, and another representing the
configuration after this communication has occurred. There would be a
transition between the states corresponding to receipt of a packet from $H_1$ to
$H_4$ at $s_4$. This model is similar to the finite state machines
used in Kinetic~\cite{kim2015kinetic} and FAST~\cite{moshref2014flow}.
However, whereas Kinetic uses uncoordinated updates,
we impose additional constraints on our ETSs which allow 
them to be implemented {\em correctly} with respect to
our consistency property. For example, we extend event-triggered
consistent updates to {\em sequences}, requiring each sequence of transitions in
the ETS to satisfy the property.
For simplicity, in this paper, we focus on finite state systems and events
corresponding to {\em packet delivery}. However, these are
not fundamental assumptions---our design extends naturally to other
notions of events, as well as infinite-state systems.

\paragraph*{Network Event Structures.}
The key challenge in implementing event-driven network programs
stems from the fact that at any time, the switches may have
different views of the global set of events that have occurred.
Hence, for a given ETS, several different updates may be
enabled at a particular moment of time, and we need a way to resolve
conflicts. We turn to the well-studied model of event
structures~\cite{winskel1987event}, which allows us to constrain transitions in two ways:
(1) {\em causal dependency}, which requires that an event $e_1$
happens before another event $e_2$ may occur, and (2) {\em compatibility},
which forbids sets of events that are in some sense
incompatible with each other from occurring in the same execution.
We present an extension called {\em network event structure} (NES), and
show how an ETS can be encoded as an NES.

\paragraph*{Locality.}
While event-driven consistent updates require immediate responses to
{\em local} events (as in the firewall), they do
not require immediate reactions to events ``at a distance.'' This is
achieved by two aspects of our definitions. 

The first defining aspect of our locality requirements involves the
happens-before (``heard-about") relation in 
event-driven consistent update. 
For example, the receipt of a packet in New
York can not immediately affect the behavior of switches in
London. Intuitively, this makes sense: requiring ``immediate" reaction to
remote events would force synchronization between switches
and buffering of packets, leading to unacceptable performance
penalties. Event-driven consistent update only requires the switches in
London to react {\em after} they have heard about the event in New York.

The second defining aspect of our locality requirements involves the
compatibility constraints in NESs. 
Suppose that New York sends packets to London and Paris, but the program
requires transitioning to a different global state based on who
received a packet first. Clearly, it
would be impossible to implement this behavior without significant
coordination. However, suppose New York and Philadelphia are
sending packets to London, and the program
requires transitioning to a different global state based on whose 
packet was received first in London. This behavior is
easily implementable since the choice is local to London. 
We use NESs to rule out non-local incompatible events---specifically,
we require that incompatible events must occur at the same switch.

Our approach gives consistency guarantees even when an event occurs at
a switch different from the one 
  that will be updated. The change will not happen
  ``atomically" with the event that triggered it, but (a) every
  packet is processed by a single configuration, and (b) the
  configuration change occurs as dictated by event-driven consistent
  update (happens-before) requirements. We show that these
  requirements can be implemented with minimal performance penalty.

Locality issues are an instance of the tension between {\em consistency} and {\em availability}
in distributed systems, which motivates existing
SDN languages to favor availability (avoiding expensive synchronization
and packet buffering) over consistency (offering strong guarantees when state
changes). 
We demonstrate that it is possible to provide the same level
of availability as existing systems, while providing a natural
consistency condition that 
is powerful enough to build many applications. We also show that
weakening the locality requirement forces us to weaken
availability.

Overall, we present a new abstraction based on (i)
a notion of causal consistency requiring that events are
propagated between nodes, (ii) per-packet consistency governing
how packets are forwarded through the network, and (iii) locality
requirements. We believe this is a powerful combination that is a natural fit for building many applications.

\paragraph*{Implementing Network Programs.}
NESs also provide a natural formalism for guiding
an implementation technique for stateful programs. Intuitively, we need switches that can
record the set of events that have been seen locally, make decisions
based on those events, and transmit events to other
switches. Fortunately, in the networking industry there is a trend
toward more programmable data planes: mutable state is already
supported in most switch ASICs (e.g. MAC learning tables) and is also
being exposed to SDN programmers in next-generation platforms such as
OpenState~\cite{bianchi2014openstate} and
P4~\cite{bosshart2014p4}. Using these features, we can implement
an NES as follows.

\begin{compactenum}
 \item Encode the sets of events contained in the NES
   as flat tags that can be carried by packets and tested on
   switches.
\item Compile the configurations contained in the NES
  to a collection of forwarding tables.
\item Add ``guards" to each configuration's forwarding rules
  to explicitly test for the tag enabling the configuration.
\item Add rules to ``stamp'' incoming packets with tags corresponding
  to the current set of events.
\item Add rules to ``learn'' which events have happened by reading
  tags on incoming packets and adding the tags in the local state to
  outgoing packets, as required to implement the happens-before
  relation.
\end{compactenum}
In this paper, we prove that a system implemented
in this way {\em correctly} implements an NES.

\paragraph*{Evaluation.}

To evaluate our design, we built a prototype of the system
described in this paper.$\dagger$
\blfootnote{$\dagger$~The PLDI 2016 Artifact Evaluation Committee (AEC) found that our prototype system %
``met or exceeded expectations." }%
We have
used this to build a number of event-driven
network applications:
\begin{inparaenum}
\item[(a)] a stateful firewall, which we have already described; 
\item[(b)] a learning switch that floods packets going to unknown
  hosts along a spanning tree, but uses point-to-point forwarding for
  packets going to known hosts; 
\item[(c)] an authentication system that initially blocks incoming traffic, but
  allows hosts to gain access to the internal network by sending
  packet probes to a predefined sequence of ports;
\item[(d)] a bandwidth cap that disables access to an external
  network after seeing a certain number of packets; and
\item[(e)] an intrusion detection system that allows all traffic until
  seeing a sequence of internal hosts being contacted in a suspicious
  order.
\end{inparaenum}
We have also built a synthetic application that forwards
packets around a ring topology, to evaluate update scalability. We developed
these applications in an extended version of NetKAT which we call {\em Stateful NetKAT}.
Our experiments show that our implementation technique
provides competitive performance on several important metrics while
ensuring important consistency properties. We draw several conclusions.
\begin{inparaenum}
\item[(1)] Event-driven consistent update allow programmers to easily
write real-world network applications and get the correct behavior,
whereas approaches relying only on uncoordinated consistency
guarantees do not.
\item[(2)] The
performance overhead of maintaining state and manipulating tags
(measured in bandwidth) is within 6\% of an implementation
that uses only uncoordinated update.
\item[(3)] There is an
optimization that exploits common structure in rules across 
states to reduce the number of rules installed on
switches.
In our experiments, a basic heuristic version of this
optimization resulted in a 32-37\% reduction in the number of rules
required on average.
\end{inparaenum}

\paragraph{Summary.}  
Our main contributions are as follows.
\begin{compactitem}
\item We propose a new semantic correctness condition for dynamic network programs called
  \emph{event-driven consistent update} that balances the need for
  immediate response with the need to avoid costly synchronization and
  buffering of packets.
Our consistency property
\xmnote{\FiveStar}{Q7}%
generalizes the guarantees offered by consistent updates, and
is as strong as possible without sacrificing availability.
\item We propose \emph{network event structures} to capture causal dependencies and compatibility
  between events, and show how to implement these using SDN functionality.
\item We describe a compiler based on a stateful extension of
  NetKAT, and present optimizations that reduce the overhead of
  implementing such stateful programs.
\item We conduct experiments showing that our approach gives
  well-defined consistency guarantees, while avoiding expensive
  synchronization and packet buffering.
\end{compactitem}

\vspace{2pt}
\noindent
The rest of this paper is structured as follows:
\S\ref{sec:consistency} formalizes event-driven consistent updates;
\S\ref{sec:model} defines event transition systems, network event
structures, and Stateful NetKAT; \S\ref{sec:implementing} describes
our implementation; and \S\ref{sec:eval} presents experiments.
We discuss related/future work in \S\ref{sec:related}-\ref{sec:discuss},
and conclude in \S\ref{sec:conclusion}.

\section{Event-Driven Network Behavior}
\label{sec:consistency}

This section presents our new consistency model for stateful network
programs: {\em event-driven consistent update}.

\paragraph*{Preliminaries.}
A {\em packet} $\pkt$ %
is a record of fields $\{f_1; f_2;\allowbreak \cdots;\allowbreak f_n\}$,
where fields $f$ represent properties such as source and
destination address, protocol type, etc. The (numeric) values of fields
are accessed via the notation $\pkt.f$, and field updates are denoted $\pkt[f \gets n]$.
A {\em switch}
$\switch$ is a node in the network with one or more {\em ports}
$\pt$. A {\em host} is a switch that can be a source or a sink of
packets. A {\em location} $l$ is a switch-port pair $n{:}m$. Locations may
be connected by (unidirectional) physical links $(l_{src},l_{dst})$ in the topology.

Packet forwarding is dictated by a {\em network
  configuration} $C$. A {\em located packet}
$\lpt=(\pkt,\switch,\port)$ is a tuple consisting of a packet and a location
$\switch{:}\port$. We model $C$ as a
relation on located packets: if $C(\lpt,\lpt')$, then the network maps
$\lpt$ to $\lpt'$, possibly changing its location and rewriting some
of its fields.
Since $C$ is a relation, it allows multiple output packets to be generated from a single input.
{In a real network, the configuration only forwards packets between ports within each individual switch,
but for convenience, we}
assume that our $C$ also captures link behavior (forwarding between switches),
i.e. $C((\pkt,n_1,m_1),(\pkt,n_2,m_2))$ holds for each link $(n_1{:}m_1,n_2{:}m_2)$.
We refer to a sequence of located packets that starts at a host and can be
produced by $C$ as a {\em packet trace}, 
using $\PTraces(C)$ to denote the set of all such packet traces.
We let $\mathcal{C}$ be the set of all configurations.

Consider a tuple $\ntr = (\lpt_0\lpt_1\cdots,T)$, where the first component is a sequence of located
packets, and each $t \in T$ is an increasing sequence of indices corresponding to
located packets in the sequence.
We call such a tuple a {\em network trace} if and only if the following conditions hold:
\begin{compactenum}
\item for each $\lpt_j$, we have $j \in t$ for some $t \in T$, and
\item for each $t = (k_0 k_1 \cdots) \in T$,
$\lpt_{k_0}$ is at a host, and $\exists C \in \mathcal{C}$ such that $C(\lpt_{k_i},\lpt_{k_{i+1}})$ holds for all $i$, and
\item if we consider the graph $G$ with nodes $\{k : (\exists t\in T : k \in t)\}$ and
edges $\{ (k_i,k_{i+1}) : (\exists t\in T : t = k_0 k_1 \cdots k_i k_{i+1} \cdots)\}$,
then $G$ is a family of {\em trees} rooted at $K = \{k_0 : (\exists t \in T : t = k_0 \cdots)\}$.
\end{compactenum}
We will use $\ntr{\downarrow}k$ to denote
the set $\{t \in T : k \in t\}$, and when $t = (k_0 k_1 \cdots) \in T$, we
can use similar notation $\ntr{\downarrow}t$ to denote the packet trace $\lpt_{k_0} \lpt_{k_1} \cdots$.
Intuitively, we have defined a network trace to be an interleaving of
these packet traces
{(the packet traces form the family of {\em trees}
because, as previously mentioned, the configuration allows multiple
output packets from a single input packet).}
Ultimately, we will introduce a consistency definition that
dictates which interleavings of packet traces are correct.

We now define how the network changes its configuration in response to
events.  An {\em event} $e$ is a tuple $(\varphi,\switch,\port)_{\eid}$,
where $\eid$ is an (optional) event identifier and $\varphi$ is a first-order formula over
fields. Events model the arrival of a packet satisfying $\varphi$ (denoted $\pkt \models \varphi$) at
location $\switch{:}\port$.  Note that we could have other types of events%
---anything that a switch can
detect could be an event---but for simplicity, we focus on packet
events.
We say that a located packet $\lpt = (\pkt,\switch',\port')$ {\em matches}
an event $e=(\varphi,\switch,\port)$ (denoted by $\lpt \models
e$) if and only if $\switch=\switch' \land \port=\port' \land \pkt \models \varphi$.

\begin{definition}[Happens-before relation $\prec_\ntr$]
Given a network trace $\ntr = (\lpt_0 \lpt_1 \cdots, T)$, the
happens-before relation $\prec_\ntr$ is the least partial order on
located packets that
\begin{compactitem}
\item respects the total order induced by $\ntr$ at switches, i.e.,
  $\forall i,j : \lpt_i \prec \lpt_j \Leftarrow i<j \land \lpt_i=
  (\pkt,\switch,\port) \land \lpt_j= (\pkt',\switch,\port')$, and
\item respects the total order induced by $\ntr$ for each packet,
  i.e., $\forall i,j : \lpt_i \prec \lpt_j \Leftarrow i<j \land
  \exists t \in T : i \in t \land j \in t$.
\end{compactitem}
\end{definition}

\paragraph*{Event-Driven Consistent Update.}
In Section \ref{sec:introduction}, we informally defined an
event-driven consistent update
as a triple $C_i \xrightarrow{e} C_f$ consisting of an initial configuration $C_i$,
event $e$, and final configuration $C_f$. 
Here, we formalize that definition in a way that describes
{\em sequences} of events and configurations (in the single-event case,
this formal definition is equivalent to the informal one).
We denote an {\em event-driven consistent update} as a pair $(U,\mathcal{E})$, where
$U$ is a sequence
$C_0 \xrightarrow{e_0} C_1 \xrightarrow{e_1} \cdots
\xrightarrow{e_{n}} C_{n+1}$,
and $\{e_0,\cdots,e_n\} \subseteq \mathcal{E}$.

Let $\ntr = (\lpt_0 \lpt_1 \cdots, T)$ be a network trace.   %
Given an event-driven consistent update 
$(U,\mathcal{E})$, 
we need the indices where the events from $U$ first
occurred. Specifically, we wish to find the 
sequence $k_0,\cdots,k_n$ where $\lpt_j$ does not match any $e\in\mathcal{E}$
for any $j > k_n$, and the
following properties hold for all $0 \leq i \leq n$ (assuming $k_{(-1)} = -1$ for
convenience): 
\begin{compactitem}
\item $k_i > k_{i-1}$, and
\item $\lpt_{k_i}$ matches $e_i$, and for all $j$, 
           if $k_{i-1} < j < k_i$ then $\lpt_{j}$ does not match $e_i$
           (i.e., $k_i$ is the first occurrence of $e_i$ after the
           index $k_{i-1}$), and
\item $\exists t \in \ntr{\downarrow}k_i$ such that $t$ is
in $\PTraces(C_{i})$ (intuitively, the event $e_i$ can be triggered only  by
a packet processed in the immediately preceding configuration). 
\end{compactitem}
If such a sequence exists, it is unique, and we denote it by $FO(\ntr,U)$, shorthand for ``first occurrences."

\begin{definition}[Event-driven consistent update correctness]
A network trace $\ntr=(\lpt_0 \lpt_1 \cdots, T)$ is {\em correct} with respect to an event-driven consistent
update $U = C_0 \xrightarrow{e_0} C_1 \xrightarrow{e_1} \cdots
\xrightarrow{e_{n}} C_{n+1}$,
if $FO(\ntr,U) = k_0,\cdots,k_n$ exists, and for all $0 \leq i \leq n$,
the following holds for each packet trace $\ntr{\downarrow}t  = \lpt_0' \lpt_1' \cdots$ where $t\in T$: 
\begin{compactitem}
\item
  $\ntr {\downarrow} t$ is in $\PTraces(C)$ for some $C\in\{C_0,\cdots,C_{n+1}\}$
  (packet is processed entirely by one configuration), and
\item if $\forall j : \lpt_j' \prec \lpt_{k_i}$, then 
  $\ntr  {\downarrow} t$ is in $\PTraces(C)$ for some $C\in\{C_0,\cdots,C_i\}$
  (the packet is processed entirely in a preceding configuration), and 
\item if $\forall j : \lpt_{k_i} \prec \lpt_j'$, then $\ntr
  {\downarrow} t$ is in $\PTraces(C)$ for some $C\in\{C_{i+1},\cdots,C_{n+1}\}$
  (the packet is
  processed entirely in a following configuration).
\end{compactitem}
\end{definition}

\begin{figure}[t]
\centering
\includegraphics[trim = 0.0in 0.0in 0.0in 0.0in, clip,width=0.55\linewidth]{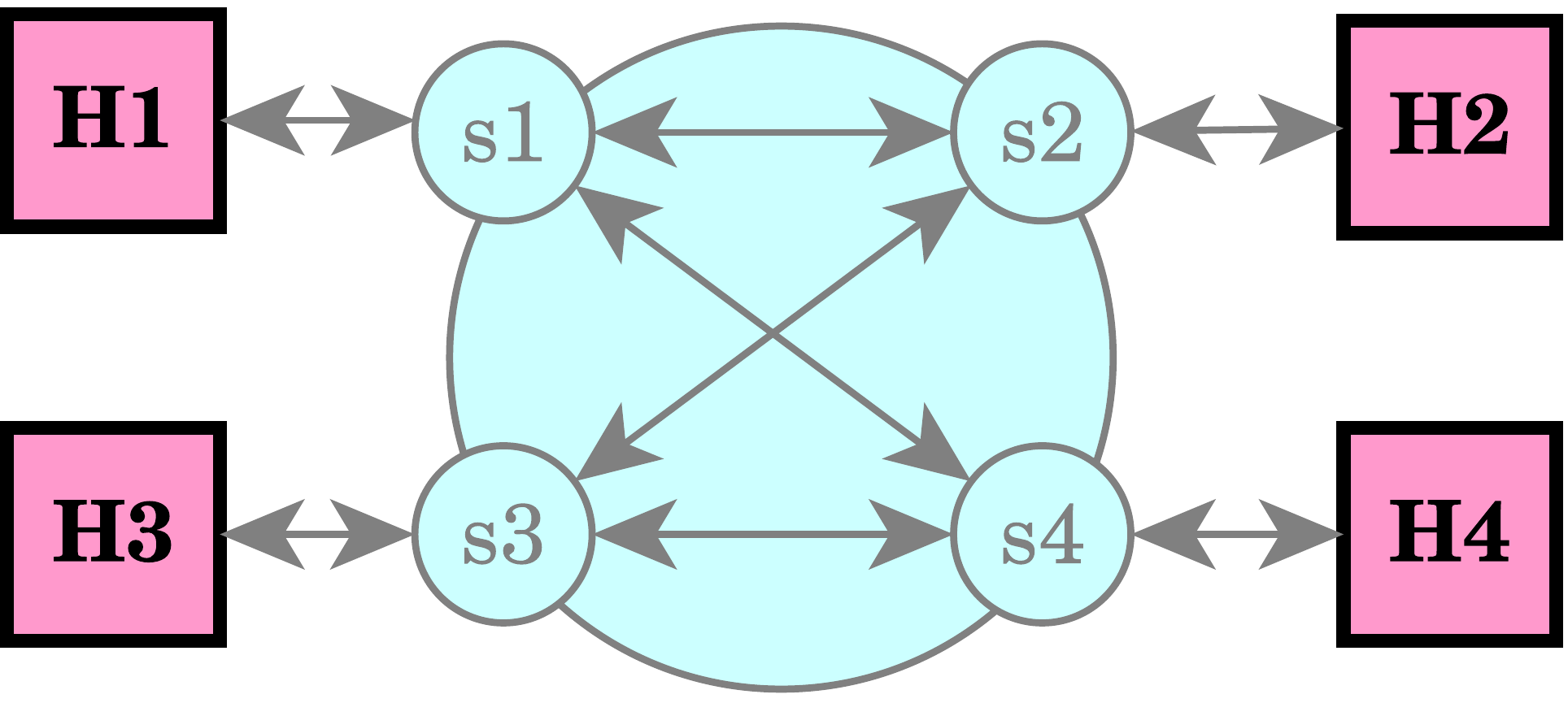}
\caption{Example topology with four switches and hosts.}
\label{fig:statefulIntro4nodes}
\end{figure}

To illustrate, consider Figure~\ref{fig:statefulIntro4nodes}. We
describe an update $C_i \xrightarrow{e} C_f$. In the initial
configuration $C_i$, the host $H_1$ can send packets to $H_2$, but not
vice-versa. In the final configuration $C_f$, traffic from $H_2$ to
$H_1$ is allowed. Event $e$ models the arrival to $s_4$ of a packet
from $H_1$ (imagine $s_4$ is part of a distributed firewall). Assume
that $e$ occurs, and immediately afterwards, $H_2$ wants to send a
packet to $s_1$. Can $s_2$ drop the packet (as it would do in
configuration $C_i$)? Event-driven consistent updates allow this, as
otherwise we would require $s_2$ to react immediately to the event at
$s_4$, which would be an example of action at a distance. Formally,
the occurrence of $e$ is not in a happens-before relation with the
arrival of the new packet to $s_2$. On the other hand, if e.g. $s_4$
forwards some packets to $s_1$ and $s_2$ before the new packet from
$H_2$ arrives, $s_1$ and $s_2$ would be required to change their
configurations, and the packet would be allowed to reach $H_1$.

\paragraph*{Network Event Structures.}
As we have seen, event-driven consistent updates specify how the network should behave
during a sequence of updates triggered by events, but additionally, we want the ability to
capture constraints
between the events themselves. For example, we might wish to say that $e_2$ can
only happen after $e_1$ has occurred, or that $e_2$ and $e_3$ cannot
both occur in the same network trace.

To model such constraints, we turn to the {\em event structures} model
introduced by Winskel~\cite{winskel1987event}.
Intuitively, an event structure endows a set of events $\mathcal{E}$
with
\begin{inparaenum}
\item[(a)] a {\em consistency predicate} ($con$) specifying which events
  are allowed to occur in the same sequence, and 
\item[(b)] an {\em enabling relation} ($\vdash$) specifying a (partial)
  order in which events can occur. 
\end{inparaenum}
This is formalized in the following definition (note that we use $\subseteq_{\fin}$ to
mean ``finite subset," and $\mathcal{P}_{\fin}(X) = \{Y : Y \subseteq_{\fin} \mathcal{P}(X)\}$).

\begin{definition}[Event structure]
\label{def:event_struct}
An event structure is a tuple $(\mathcal{E},con,\vdash)$ where:
\begin{compactitem}
\item $\mathcal{E}$ is a set of events, 
\item $con : (\mathcal{P}_{\fin}(\mathcal{E}) \rightarrow Boolean)$ is a
  consistency predicate that satisfies $con(X) \land Y \subseteq X
  \implies con(Y)$, %
\item $\vdash\; : (\mathcal{P}(\mathcal{E}) \times \mathcal{E}
  \rightarrow Boolean)$ is an enabling relation that satisfies $(X
  \vdash e) \land X \subseteq Y \implies (Y \vdash e)$.
\end{compactitem}
\end{definition}

\noindent An event structure can be seen as defining a transition
system whose states are subsets of $\mathcal{E}$
that are
consistent and reachable via the enabling relation. We refer to such a
subset an as an {\em event-set} (called ``configuration"
in~\cite{winskel1987event}).

\begin{definition}[Event-set of an event structure]
\label{def:event_struct_config}
Given an event structure $N = (\mathcal{E},con,\vdash)$, an {\em
  event-set} of $N$ is any subset $X \subseteq \mathcal{E}$ which is:
\begin{inparaenum}
\item[(a)] consistent: $\forall Y \subseteq_\fin X$, $con(Y)$ holds, and
\item[(b)] reachable via the enabling relation: for each $e \in X$, there
  exists $e_0,e_1,\cdots,e_n \in X$ where $e_n = e$ and $\emptyset \vdash \{e_0\}$ and
  $\{e_0,\cdots,e_{i-1}\} \vdash e_i$ for all $1 \leq i \leq n$. 
\end{inparaenum}
\end{definition}

We want to be able to specify which network configuration should be
active at each event-set of the event structure.  Thus, we need the
following extension of event structures.

\begin{definition}[Network event structure (NES)]
A {\em network event structure} is a tuple
$(\mathcal{E},con,\vdash,g)$ where
$(\mathcal{E},con,\vdash)$ is an event structure, and
$g : (\mathcal{P}(\mathcal{E}) \rightarrow \mathcal{C})$ maps each
event-set of the event structure to a network configuration.
\end{definition}

\paragraph*{Correct Network Traces.}
\label{sec:correct_trace}
We now define what it means for a network trace $\ntr$
to be {\em correct} with respect to an
NES $N=(\mathcal{E},con,\vdash,g)$. We begin by constructing a
sequence $S$ of events that is allowed by $N$. %
A sequence $S = e_0e_1 \cdots e_n$ is {\em allowed by} $N$, if 
$\emptyset \vdash \{e_0\} \land con(\{e_0\})$, and
$\forall 1 \leq i \leq n: (\{e_0,e_1,\cdots,e_{i-1}\} \vdash e_i \allowbreak\land\allowbreak con(\{e_0,e_1,\cdots,e_{i}\}))$.

Intuitively, we say that $\ntr$ is correct if there is a sequence of events allowed by $N$
which would cause $\ntr$ to satisfy the event-driven consistent update condition.

\begin{definition}[Correct network trace]
\label{def:corr_trace}%
Let $\mathcal{S}$ be the set of all sequences allowed by $N$. 
Formally, a network trace $\ntr=(\lpt_0\lpt_1\cdots,T)$ is {\em correct} with respect to $N$ if
\begin{compactitem}
\item no $\lpt_j$ matches any $e \in \mathcal{E}$, and %
for all packet traces $\ntr{\downarrow}t$ where $t \in T$, 
we have $\ntr{\downarrow}t$ is in $\PTraces(g(\emptyset))$, or
\item there exists some $e_0 e_1 \cdots e_n \in \mathcal{S}$ such that
$\ntr$ is correct with respect to event-driven consistent update $(g(\emptyset) \xrightarrow{e_0} g(\{e_0\}) \xrightarrow{e_1} \cdots \xrightarrow{e_n} g(\{e_0,\cdots,e_n\}),\mathcal{E})$.
\end{compactitem}
\end{definition}

\paragraph*{Locality Restrictions for Incompatible Events.}
\label{sec:locrestr}
We now show how NESs can be used to impose reasonable locality
restrictions. A set of events $E$ is called {\em
  inconsistent} if and only if $con(E)$ does not hold.  We use the term {\em
  minimally-inconsistent} to describe inconsistent sets where all proper subsets are
not inconsistent.
An NES $N$ is called {\em locally-determined} if and only if for each of its
minimally-inconsistent sets $E$, all events in $E$ happen
at the same switch (i.e., $\exists \switch \forall e_i \in E : e_i = (\varphi_i,\switch,\port_i)$).
To illustrate the need for the locally-determined property, let us
consider the following two programs, $P_1$ and $P_2$.

\begin{compactitem}
\item
{\em Program $P_1$}: Recall that two events are inconsistent if either of
them can happen, but both cannot happen in the same
execution. Consider the topology shown in
Figure~\ref{fig:statefulIntro4nodes} and suppose this program requires
that $H_2$ and $H_4$ can both receive packets from $H_1$, but only the
first one to receive a packet is allowed to respond. There will be two
events $e_1$ and $e_2$, with $e_1$ the arrival of a packet from $H_1$
at $s_2$, and $e_2$ the arrival of a packet from $H_1$ at $s_4$. These
events are always enabled, but the set $\{e_1,e_2\}$ is not
consistent, i.e. $con(\{e_1,e_2\})$ does not hold. This models the
fact that at most one of the events can take effect.  These events
happen at different switches---making sure that {\em at most one} of
the events takes effect would necessitate information to be propagated
instantaneously ``at a distance.'' In implementations, this would
require using inefficient mechanisms (synchronization and/or
packet buffering). Our locality restriction is a clean
condition which ensures that the NES is efficiently implementable.

\item
{\em Program $P_2$}: Consider a different program
where $H_2$ can send traffic to one of the two hosts $H_1, H_3$ that
sends it a packet first. The two events (a packet from $H_1$ arriving
at $s_2$, and a packet from $H_3$ arriving at $s_2$) are still
inconsistent, but inconsistency does not cause problems in this case,
because both events happen at the same switch (the switch can
determine which one was first).
\end{compactitem}

\propose{
\noindent
In contrast to our approach,
\xmnote{\FiveStar}{Q5}%
an uncoordinated update approach improperly handles locality issues, mainly because it
does not guarantee {\em when} the
  configuration change occurs.
Consider the program $P_1$ again,
and consider the (likely) scenario where events $e_1$ and $e_2$ happen
nearly simultaneously.
  In an uncoordinated approach, this could result in switch $s_2$
  hearing about $e_1,e_2$ (in that order),
  and $s_4$ hearing about $e_2,e_1$ (in that order),
  meaning the two switches would have conflicting
  ideas of which event was ``first" (i.e. the switches would be in conflicting states,
  and this conflict cannot be resolved).
  In our implementation, we would require $e_1 $
  and $e_2$ to occur at the same switch, guaranteeing
  that we never see such a conflicting mix of states.
}%

\paragraph*{Strengthening Consistency.}
\renewcommand*{\proofname}{Proof Sketch}%
We now show that strengthening the consistency conditions imposed by
NESs would lead to lower availability, as it would lead to the need
for expensive synchronization, packet buffering, etc. First, we will
try to remove the locally-determined condition, and second, we will try to obtain a
strengthened consistency condition. The proof of the following theorem
is an adaptation of the proof of the CAP theorem~\cite{B00}, as
presented in~\cite{GL12}. The idea is that in asynchronous network
communication, a switch might need to wait arbitrarily long to hear
\xmnote{\FiveStar}{Q14}%
about an event.
\begin{lemma}
\label{lem:local}
In general, it is impossible to implement an NES that does not have 
the locally-determined condition while guaranteeing that switches
process each packet within an {a priori} given time bound.
\end{lemma}
\begin{proof}
Consider a simple NES, with event sets
$\emptyset$, $\{e_1\}$, $\{e_2\}$, and where $\{e_1\}$ and $\{e_2\}$ are both enabled from 
$\emptyset$. Assume that $con(\{e_1,e_2\})$ does not hold,
and that $e_1$ can happen at switch $A$ and $e_2$ 
can happen at switch $B$ (i.e., the locally-determined condition does not
hold). 

Because the communication is asynchronous, there is no a priori
bound on how long the communication between switches can take.
When a packet $p$ that matches $e_2$ arrives at the switch $B$, the switch must
distinguish the following two cases: 
\begin{inparaenum}
\item[(\#1)] event $e_1$ has occurred at $A$ (and thus $p$ does not cause
  $e_2$), or
\item[(\#2)] event $e_1$ has not occurred at $A$ (and thus $p$ causes 
  $e_2$).
\end{inparaenum}
No matter how long $B$ waits, it cannot distinguish these two cases,
and hence, when a packet that matches $e_2$ arrives to $B$, the switch
$B$ cannot
correctly decide whether to continue as if $e_2$ has happened. It has
the choice to either eventually decide (and risk the wrong decision),
or to buffer the packet that matches $e_2$.  
\end{proof}

\noindent We now ask whether we can strengthen the event-driven
consistent update definition. We define {\em strong update} as an
update $C_1 \xrightarrow{e} C_2$ such that immediately after 
$e$ occurred, the network processes all incoming packets in $C_2$. We
obtain the following lemma by the same reasoning as the previous one.
\begin{lemma}
\label{lem:packet}
In general, it is impossible to implement strong updates and guarantee that
switches process each packet within an {a priori} given time bound.
\end{lemma}
\begin{proof}
Let $A$ be the switch where $e$ can happen, and let $B$ be a switch
on which the configurations $C_1, C_2$ differ. For $A$ and $B$,
the same argument as in the previous lemma shows that $B$ must
either risk the wrong decision on whether to process packets using
$C_1$ or $C_2$, or buffer packets.
\end{proof}

\renewcommand*{\proofname}{Proof}%

\section{Programming with Events}
\label{sec:model}
The correctness condition we described in the previous section
offers useful application-level guarantees to network programmers.
 At a high level, the programmer is freed from thinking about 
  interleavings of packets/events and responses to events
  (configuration updates). She can think in terms of our consistency
  model---each packet is processed in a single configuration,
  and packets entering ``after" an event will be processed in the new
  configuration (similar to causal consistency). An important consequence %
  is that the response to an event is immediate with respect to a
  given flow if the event is handled at that flow's ingress switch.

With this consistency model in mind, programmers can proceed by
specifying the desired event-driven program behavior using network
event structures. %
This section introduces an intuitive method for building NESs using
simple transition systems where nodes correspond to configurations and edges
correspond to events. We also present a network programming language
based on NetKAT that provides a compact notation for specifying both the
transition system and the configurations at the nodes.

\subsection{Event-Driven Transition Systems}
\label{subsec:ets}
\begin{definition}[Event-driven Transition System]
An {\em event-driven transition system (ETS)} is a graph $(V,D,v_0)$,
in which $V$ is a set of vertices, each labeled by a
configuration; $D \subseteq V \times V$ is a
set of edges, each labeled by an event $e$; 
and $v_0$ is the initial vertex. 
\end{definition}

\begin{figure}[t]
\centering
\includegraphics[trim = 0.0in 0.0in 0.0in 0.0in,
  clip,width=1.0\linewidth]{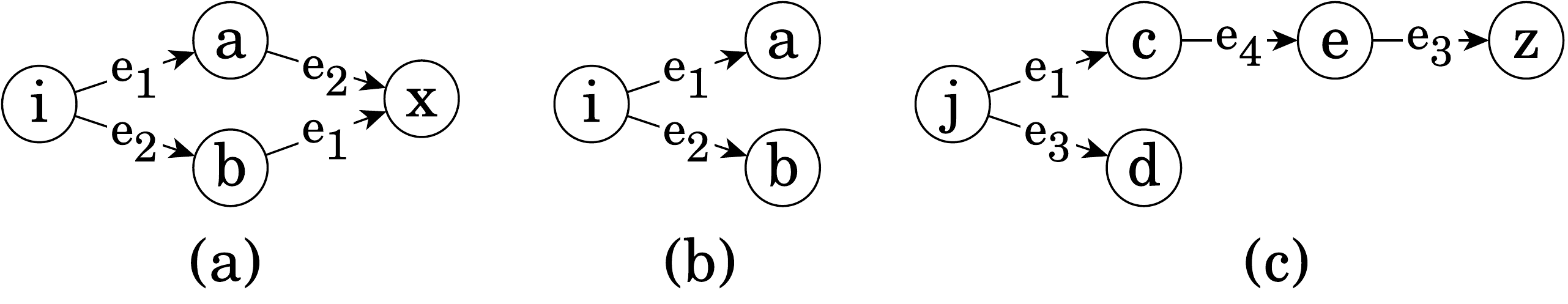}
\caption{Event-driven transition systems.}
\label{fig:ETSs}
\end{figure}

Consider the ETSs shown in Figure~\ref{fig:ETSs} (a-b). In (a), the
two events are intuitively {\em compatible}---they can happen in any order,
so we obtain a correct execution if both happen in different parts of
the network, and different switches can have a different view of the
order in which they happened. In (b), the two events are intuitively
{\em incompatible}---only one of them can happen in any particular
execution. Therefore, even if they happen nearly simultaneously, only
one of them should take an effect.
To implement this, we require the
locality restriction---we need to check whether the two events happen
at the same switch.
We thus need to distinguish between ETSs such as
(a) and (b) in Figure~\ref{fig:ETSs}, to determine where
locality restrictions must be imposed in the conversion from an ETS to
an NES.

\paragraph*{From ETSs to NESs.}
To convert an ETS to an NES, we first form the event sets
(Definition \ref{def:event_struct_config}) and then construct the enabling relation and consistency
predicate.
Given an ETS $T$, consider the
set $W(T)$ of sequences of events in $T$ from the initial node to any
vertex (including the empty sequence). For
each sequence $p \in W(T)$, let $E(p)$ be the set of events collected
along the sequence. The set $F(T) =\{ E(p) \mid p \in W(T) \}$ is our
candidate collection of event sets.
We now define conditions under which $F(T)$ gives rise to an NES.

\begin{compactenum}
\item We require that each set $E$ in $F(T)$ must correspond to
  exactly one network configuration. This holds if all paths in
  $W(T)$ corresponding to $E$ end at states labeled with the same
  configuration.

\item We require that $F(T)$ is {\em finite-complete}, i.e. for any
  sets $E_1,E_2,\cdots,E_n$ where each $E_i \in F(T)$, if there is a set
  $E' \in F(T)$ which contains
  every $E_i$ (an upper bound for the sets $E_i$), then the set $E_{lub} = \cup_i E_i$ (the least upper
  bound for the $E_i$) must also be in $F(T)$.  For example, consider the
  ETS in Figure~\ref{fig:ETSs}(c), which violates this condition since
  the event-sets $E_1=\{e_1\}$ and $E_2=\{e_3\}$ are both subsets of
  $\{e_1, e_4, e_3\}$, but there is no event-set of the form $E_1 \cup
  E_2 = \{e_1, e_3\}$.
\end{compactenum}

\noindent In~\cite{winskel1987event}, such a collection $F(T)$ is
called a {\em family of configurations}. Our condition (2) is
condition (i) in Theorem 1.1.9 in~\cite{winskel1987event} (conditions
(ii)-(iii) are satisfied by construction).

Given an ETS $T$, it is not difficult to confirm the above conditions
\xmnote{\FiveStar}{Q8}%
statically.
They can be checked using
straightforward graph algorithms, and any problematic vertices or
edges in $T$ can be indicated to the programmer.
The development of {\em efficient} checking algorithms is left
for future work.

We build the $con$ and $\vdash$ relations of an NES from
the family $F(T)$, using Theorem
1.1.12. of~\cite{winskel1987event}. Specifically, predicate $con$ can be defined by
declaring all sets in $F(T)$ as consistent, and for $\vdash$, we take the
smallest relation satisfying the constraints $\emptyset \vdash e \iff
\{e\} \in F(T)$ and $X' \vdash e \iff (X' \in con)\; \land ((X' \cup
\{e\}) \in F(T) \lor (\exists X \subseteq X' : X \vdash e))$.

After obtaining an NES, deciding whether it satisfies the locality
restriction is easy: we check whether the NES is locally determined
(see Section~\ref{sec:locrestr}), verifying for each
minimally-inconsistent set that the locality restriction holds.
Again, we leave the {\em efficiency} of this check for future work.

\paragraph{Loops in ETSs.}
If there are loops in the ETS $T$, the previous definition needs to be
slightly modified, because we need to ``rename" events encountered
multiple times in the same execution. This gives rise to an NES where
each event-set is finite, but the NES itself might be infinite (and
thus can only be computed lazily). If we have the ability to
store and communicate unbounded (but finite) event-sets in the network
runtime, then no modifications are needed to handle infinite NESs in
the implementation (which is described in Section \ref{sec:implementing}).
Otherwise, there are various correct overapproximations we could use,
such as computing the
strongly-connected components (SCCs) of the ETS, enforcing the
locality restriction on events in each (non-singleton) SCC, and
requiring the implementation to attach timestamps on occurrences of
events in those SCCs. For simplicity of the presentation, we will consider only loop-free ETSs
in this paper.

\subsection{Stateful NetKAT}

NetKAT \cite{anderson2014netkat} is a domain-specific language for
specifying network behavior. It has semantics based on Kleene
Algebra with Tests (KAT), and a sound and complete equational theory
that enables formal reasoning about programs. Operationally, a NetKAT
program behaves as a function which takes as input a single packet,
and uses tests, field-assignments, sequencing, and union to produce a
set of ``histories'' corresponding to the packet's traces.

Standard NetKAT does not support mutable {\em state}. Each packet is
processed in isolation using the function described by the
program. In other words, we can use a standard NetKAT program for specifying
individual network configurations, but not event-driven configuration
changes.  We describe a stateful variant of NetKAT which allows us to
compactly specify a {\em collection} of network configurations, as
well as the event-driven relationships between them (i.e. an ETS).  This variant
preserves the existing equational theory of the individual static configurations
(though it is not a KAT itself), but also allows packets to affect
processing of future packets via assignments to (and
tests of) a global {\em state}.
The syntax of Stateful NetKAT is shown in Figure
\ref{fig:st_netkat_syntax}.
A Stateful NetKAT program is a {\em command}, which can be:
\begin{compactitem}
\item a {\em test}, which is a formula over packet header fields
  (there are special fields $\kw{sw}$ and $\kw{pt}$ which test the
  switch- and port-location of the packet respectively),
\item a {\em field assignment} $x{\leftarrow}n$, which modifies the
  (numeric) value stored in a packet's field,
\item a {\em union} of commands $p + q$, which unions together the
  packet-processing behavior of commands $p$ and $q$,
\item a command {\em sequence} $p\,;q$, which runs packet-processing
  program $q$ on the result of $p$,
\item an {\em iteration} $p*$, which is equivalent to $\kw{true} + p + (p\,;
  p) + (p\,; p\,; p) + \cdots$,
\item or a {\em link definition}
  $(n_1{:}m_1){\rightarrowtriangle}(n_2{:}m_2)$, which forwards a
  packet from port $m_1$ at switch $n_1$ across a physical link to port $m_2$ at switch
  $n_2$.
\end{compactitem}

\plstxset{or skip=2.5pt}
\begin{figure}[t]
\small
\centering
\begin{plstx}
*(packet field name): f [\in] \mathit{Field} \\
*(numeric value): n [\in] \mathbb{N} \\
(modifiable field): x ::= f | \kw{pt} \\
(test): a,b ::=
\kw{true} | \kw{false} | x=n | \kw{sw}{=}n | \kw{state}(n) = n | a\lor b | a\land b | \lnot a \\
(command): p,q ::=
a | x\leftarrow n | p+q | p \;;\, q | p *
| (n{:}n) \rightarrowtriangle (n{:}n)
| (n{:}n) \rightarrowtriangle (n{:}n) \rightarrowtriangle \langle \kw{state}(n) \leftarrow n \rangle \\
\set{or skip=1pt,or=\,}
\set{or=\vert}
\end{plstx}
\caption{Stateful NetKAT: syntax.}
\label{fig:st_netkat_syntax}
\end{figure}

\noindent
The functionality described above is also provided by standard NetKAT \cite{smolka2015}.
The key distinguishing feature of our {\em Stateful} NetKAT is a special
global vector-valued variable called $\kw{state}$, which allows the
programmer to represent a collection of NetKAT programs. The function
shown in Figure \ref{fig:state_static_semantics} gives the standard
NetKAT program $\llbracket p \rrbracket_{\vec{k}}$ corresponding to
each value $\vec{k}$ of the state vector (for conciseness, we only
show the non-trivial cases).  We can use the NetKAT compiler
\cite{smolka2015} to generate forwarding tables (i.e. configurations)
corresponding to these, which we denote $C(\llbracket p
\rrbracket_{\vec{k}})$.

\begin{figure}[t]
\footnotesize
\begin{displaymath}
\begin{array}{rcl}
\llbracket \kw{state}(m){=}n \rrbracket_{\vec{k}} & \triangleq & 
\begin{cases}
\llbracket \kw{true} \rrbracket_{\vec{k}} & \text{if}\; \vec{k}(m){=}n \\
\llbracket \kw{false} \rrbracket_{\vec{k}} & \text{otherwise}
\end{cases} \vspace{2pt} \\
\llbracket (a{:}b) \rightarrowtriangle (c{:}d) \rightarrowtriangle \langle \kw{state}(m) \leftarrow n \rangle\rrbracket_{\vec{k}} & \triangleq &
\llbracket (a{:}b) \rightarrowtriangle (c{:}d) \rrbracket_{\vec{k}} \\
\end{array}
\end{displaymath}
\caption{Stateful NetKAT: extracting NetKAT Program (state $\vec{k}$).}
\label{fig:state_static_semantics}
\end{figure}

\begin{figure*}[t]
\small
\centering
\fbox{
\begin{minipage}{0.98\linewidth}
\begin{tabular}{l | l}
\begin{minipage}[b]{0.49\linewidth}
\vspace{-10pt}
\begin{displaymath}
\begin{array}{rcl}
\llrrparen{f \oeq n}_{\vec{k}}\; \varphi & \triangleq & (\{ \}, \{\varphi \land f{\oeq}n\}) \\
\llrrparen{\kw{sw} \oeq n}_{\vec{k}}\; \varphi & \triangleq & \llrrparen{\kw{true}}_{\vec{k}}\; \varphi \\
\llrrparen{\kw{port} \oeq n}_{\vec{k}}\; \varphi & \triangleq & \llrrparen{\kw{true}}_{\vec{k}}\; \varphi \vspace{1pt} \\

\llrrparen{state(m) \oeq n}_{\vec{k}}\; \varphi & \triangleq & 
\begin{cases}
      \llrrparen{\kw{true}}_{\vec{k}}\; \varphi  & \text{if}\; \vec{k}(m){\oeq}n \\
      \llrrparen{\kw{false}}_{\vec{k}}\; \varphi  & \text{otherwise} \\
\end{cases} \vspace{1pt} \\

\llrrparen{f \leftarrow n}_{\vec{k}}\; \varphi & \triangleq & (\{ \}, \{(\exists f : \varphi) \land f{=}n \}) \\
\llrrparen{p + q}_{\vec{k}}\; \varphi & \triangleq &
(\llrrparen{p}_{\vec{k}}\; \varphi) \sqcup
(\llrrparen{q}_{\vec{k}}\; \varphi)
\\

\llrrparen{p \;;\, q}_{\vec{k}} \;\varphi & \triangleq & (\llrrparen{p}_{\vec{k}} \sqcdot \llrrparen{q}_{\vec{k}})\; \varphi \\
\llrrparen{p*}_{\vec{k}} \;\varphi & \triangleq & \bigsqcup_j F^j_p\; (\varphi ,\vec{k})\\

\end{array}
\end{displaymath}
\end{minipage} &
\begin{minipage}{0.49\linewidth}
\begin{displaymath}
\begin{array}{rcl}
\llrrparen{a \land b}_{\vec{k}}\; \varphi & \triangleq & \llrrparen{a \;;\, b}_{\vec{k}}\; \varphi \\
\llrrparen{a \lor b}_{\vec{k}}\; \varphi & \triangleq & \llrrparen{a + b}_{\vec{k}}\; \varphi \\

\llrrparen{\kw{true}}_{\vec{k}}\; \varphi & \triangleq & (\{ \}, \{ \varphi \}) \\
\llrrparen{\kw{false}}_{\vec{k}}\; \varphi & \triangleq & (\{ \}, \{ \}) \\

\llrrparen{\lnot \kw{true}}_{\vec{k}}\; \varphi & \triangleq & \llrrparen{\kw{false}}_{\vec{k}}\; \varphi \\
\llrrparen{\lnot \kw{false}}_{\vec{k}}\; \varphi & \triangleq & \llrrparen{\kw{true}}_{\vec{k}}\; \varphi \\
\llrrparen{\lnot (v \oeq n)}_{\vec{k}}\; \varphi & \triangleq & \llrrparen{v \oneq n}_{\vec{k}}\; \varphi \\
\llrrparen{\lnot \lnot a}_{\vec{k}}\; \varphi & \triangleq & \llrrparen{a}_{\vec{k}}\; \varphi \\
\llrrparen{\lnot (a \land b)}_{\vec{k}}\; \varphi & \triangleq & \llrrparen{\lnot a \lor \lnot b}_{\vec{k}}\; \varphi \\
\llrrparen{\lnot (a \lor b)}_{\vec{k}}\; \varphi & \triangleq & \llrrparen{\lnot a \land \lnot b}_{\vec{k}}\; \varphi \\

\end{array}
\end{displaymath}
\end{minipage}
\end{tabular}

\medskip
\hrule
\smallskip

\begin{minipage}{0.98\linewidth}
\begin{displaymath}
\begin{array}{rcl}
\llrrparen{(s_1{:}p_1) \rightarrowtriangle (s_2{:}p_2)}_{\vec{k}}\; \varphi & \triangleq & (\{ \}, \{ \varphi \}) \\
\llrrparen{(s_1{:}p_1) \rightarrowtriangle (s_2{:}p_2) \rightarrowtriangle \langle\kw{state}(m)\leftarrow n\rangle}_{\vec{k}}\; \varphi & \triangleq & (\{ (\vec{k},(\varphi,s_2,p_2),\vec{k}[m \mapsto n]) \}, \{ \varphi \}) \\
\end{array}
\end{displaymath}
\end{minipage}
\end{minipage}
}
\caption{Stateful NetKAT: extracting event-edges from state $\vec{k}$.}
\label{fig:state_semantics_conf}
\end{figure*}

\subsection{Converting Stateful NetKAT Programs to ETSs}

Now that we have the $\llbracket \cdot \rrbracket_{\vec{k}}$ function 
to extract the static configurations (NetKAT programs) corresponding to the vertices of an ETS,
we define another function $\llrrparen{\cdot}_{\vec{k}}$,
which produces the event-edges
(Figure \ref{fig:state_semantics_conf}).
This collects (using parameter $\varphi$) the conjunction of all tests seen up to
a given program location, and records a corresponding event-edge when a state assignment
command is encountered.
The function returns a tuple $(D,P)$, where $D$ is a set of event-edges, and $P$
is a set of updated conjunctions of tests.
In the figure, the $\sqcup$ operator denotes pointwise union of tuples,
i.e. $(A_1, B_1, \cdots) \sqcup (A_2, B_2, \cdots) = (A_1 \cup A_2, B_1 \cup B_2, \cdots)$.
The $\sqcdot$ operator denotes (pointwise) Kleisli composition, i.e. 
$(f \sqcdot g) \triangleq \bigsqcup\, \{ g\; y : y \in f\; x\}$, and function $F$ is as follows.
\[\begin{array}{rcl}
F^0_p\; (\varphi,\vec{k}) & \triangleq & (\{ \}, \{ \varphi \}) \\
F^{j+1}_p\; (\varphi,\vec{k}) & \triangleq & (\llrrparen{p}_{\vec{k}} \sqcdot F^j_p)\; \varphi
\end{array}\]

\noindent
The symbol variable $\oeq$ is either equality ``$=$" or inequality ``$\not=$",
and $\oneq$ is the opposite symbol with respect to $\oeq$.
Given any conjunction $\varphi$ and a header field $f$, the formula 
$(\exists f : \varphi)$ strips all predicates of the form $(f \oeq n)$ from $\varphi$.

Using $\mathit{fst}$ to denote obtaining the first element of a tuple, we can now produce the
event-driven transition system for a Stateful NetKAT program $p$ with the initial state $\vec{k}_0$:

\[
\begin{array}{rl}\mathit{ETS}(p) \triangleq & (V,D,v_0)\\
\multicolumn{2}{l}{\quad \text{where}~V \triangleq \bigcup_{\vec{k}} \{(\vec{k},C(\llbracket p \rrbracket_{\vec{k}}))\}}\\
\multicolumn{2}{l}{\qquad \text{and}~D \triangleq \mathit{fst}\,\big(\bigsqcup_{\vec{k}} \llrrparen{p}_{\vec{k}}\; \kw{true}\big)}\\
\multicolumn{2}{l}{\qquad \text{and}~v_0 \triangleq (\vec{k}_0,C(\llbracket p \rrbracket_{\vec{k}_0}))}\\
\end{array}
\]

\section{Implementing Event-Driven Programs}
\label{sec:implementing}

Next, we show one method of implementing NESs in a real SDN, and we prove that
this approach is correct---i.e., all traces followed by actual packets
in the network are correct with respect to Definition \ref{def:corr_trace} in Section~\ref{sec:correct_trace}.
At a high level, the basic idea of our implementation strategy can be
understood as follows. We assume that the switches in the network
provide mutable state that can be read and written as packets are
processed. Given an NES, we assign a tag to each event-set and compile
to a collection of configurations whose rules are ``guarded" by the
appropriate tags. %
We then add logic that (i) updates the mutable state to
record local events, (ii) stamps incoming packets with the tag for
the current event-set upon ingress, and (iii) reads the tags carried
by packets, and updates the event-set at subsequent switches.

\subsection{Implementation Building Blocks}
\label{subsec:building_blocks}

\paragraph{Static Configurations.}
The NES contains a set of network configurations that need to be
installed as flow tables on switches. In addition, we must
be able to transition to a new configuration in response to a local event. 
We do this {\em proactively}, installing all of the needed rules on
switches in advance, with each rule guarded by its configuration's ID.
This has a disadvantage of being less efficient in terms of rule-space
usage, but an advantage of allowing quick configuration
changes. In Section \ref{sec:opt}, we discuss an approach for
addressing the space-usage issue by sharing rules between configurations.
\propose{Our implementation strategy
encodes each event-set in the NES
\xmnote{\FiveStar}{Q11}%
as an integer, so a single unused packet header field (or single
register on switches) can be used. This keeps the overhead
low, even for very large programs.}

\implTraceFig%

\paragraph{Stateful Switches.}
Emerging data-plane  languages such as P4 \cite{bosshart2014p4} and
OpenState~\cite{bianchi2014openstate} are 
beginning to feature advanced functionality such as customizable
parsing, stateful memories, etc.
We assume that our switches support
\propose{(1) modifying a local register (e.g. an integer on a switch)
appropriately upon receipt of a packet,
and (2) making packet forwarding decisions based on the value of a register.}
This allows each switch to
maintain a local view of the global state. Specifically, the register
records the set of events the device knows have occurred. At any time,
the device can receive a packet (from the controller or another
device) informing it of new event occurrences, which are ``unioned" into
the local register (by performing a table lookup based on integer values).
\propose{Currently, P4 data planes support this type of functionality.}

\propose{We also assume that
\xmnote{\FiveStar}{Q13}%
the switch atomically processes each packet
in the order in which it was received. Such ``atomic" switch operations
are proposed by the ``Packet Transactions" P4 extension \cite{sivaraman2015}.
Because the P4 switch platform is attracting considerable attention
(even spawning its own highly-attended workshop),
we feel that our assumptions are realistic for the current
state-of-the-art in regards to switches.}

\paragraph{Packet Processing.}
Each packet entering the network is admitted from a host to a port
on an edge switch.  The configuration ID $j$ corresponding to the device's
view of the global state is assigned to the packet's {\em version
  number} field. The packet will processed only by $j$-guarded rules
throughout its lifetime.
Packets also carry a {\em digest} encoding the set of events the
packet has heard about so
far (i.e. the packet's view of the global state). If the packet passes
through a device which has heard about additional events, the packet's
digest is updated accordingly. Similarly, if the packet's digest
contains events not yet heard about by the device, the latter adds them to
its view of the state.
When a packet triggers an event, that event is immediately added to
the packet's digest, as well as to the state of the device where the
event was detected. The controller is then notified about the
event. Optionally (as an optimization), the controller
can periodically broadcast its view
of the global state to all switches, in order to speed up
dissemination of the state.

\subsection{Operational Model}
\label{subsec:model}

We formalize the above via operational semantics for the global
behavior of the network as it executes an NES.
Each state in Figure \ref{fig:actual_model} has the form $(Q,R,S)$,
with a controller queue $Q$, a controller $R$, and set of switches
$S$.  Both the controller queue and controller are a set of events,
and initially, $R{=}Q{=}\emptyset$.
Each switch $s\in S$ is a tuple $(n,qm_{in},E,qm_{out})$, where $n$ is
the switch ID, $qm_{in}, qm_{out}$ are the input/output queue maps
(mapping port IDs to packet queues).
Map updates are denoted $qm[m \mapsto \pkts]$.
The event-set $E$ represents a switch's view of what events have
occurred.
A packet's digest is denoted $\pkt.digest$, and the configuration
corresponding to its version number is denoted $\pkt.C$.
The rules in Figure \ref{fig:actual_model} can be summarized as follows.
\begin{compactitem}
\item {\sc In/Out}: move a packet between a host and edge port.
\item {\sc Switch}: process a packet by first adding new events from the
  packet's digest to the local state, then checking if the packet's
  arrival matches an event $e$ enabled by the NES and updating the
  state and packet digest if so, and finally updating the digest with
  other local events.
\item {\sc Link}: move a packet across a physical link.
\item {\sc CtrlRecv}: bring an event from the controller
  queue into the controller.
\item {\sc CtrlSend}: update the local state of the switches.
\end{compactitem}

\subsection{Correctness of the Implementation}
\label{sec:correctness}

We now prove the correctness of our implementation. Formally, we show
that the operational semantics generates correct traces, as defined in Section~\ref{sec:correct_trace}.

\renewcommand*{\proofname}{Proof Sketch}%

\begin{lemma}[Global Consistency]
\label{lem:consistent}
Given a locally-determined network event structure $N$, %
for an execution of the implementation
$(Q_1,R_1,S_1) (Q_2,R_2,S_2) \cdots (Q_m,R_m,\allowbreak S_m)$, the event-set
$Q_i \cup R_i$ is 
consistent for all $1 \leq i \leq m$.
\end{lemma}
\begin{proof}

We first show that if an {\em inconsistent} set $Y$ where $|Y| > 1$
satisfies the locality restriction (i.e. all of its events are
handled at the same switch), then $Y \subseteq R_i \cup Q_i$ is
not possible for any $i$ (the {\sc Switch} rule ensures that multiple events from $Y$ could not have been 
sent to the controller).

We proceed by induction over $m$, the trace length,
noting that the base case $Q_0 \cup R_0 = \emptyset$ is consistent.
Assume that the implementation
adds an $e$ (via {\sc Switch}) to some consistent event-set $Q_m \cup R_m$, producing an inconsistent set. We look
at the minimally-inconsistent set $Y \subseteq (Q_m \cup R_m \cup\{e\})$, and notice that the
locality restriction requires all events in $Y$ to be detected at the same switch,
so by the previous paragraph, we must have $|Y| \leq 1$.
This generates a contradiction, since it would mean that either $Y =\{e_0\}$
or $Y \subseteq Q_m \cup R_m$, either of which would make $Y$ consistent.
\end{proof}

\paragraph{Traces of the Implementation.}
Note that we can readily produce the network trace (Section \ref{sec:consistency})
that corresponds to an implementation trace, since a single packet $\pkt$ is processed at each step of
Figure \ref{fig:actual_model}.
We now present the main result of this
section---executions of the implementation correspond to correct network
  traces (Definition \ref{def:corr_trace}).

\begin{theorem}[Implementation Correctness]
\label{thm:correct}
For an NES $N$, and an execution
$(Q_1,R_1,S_1) (Q_2,R_2,S_2) \cdots (Q_m,\allowbreak R_m,\allowbreak S_m)$ of the implementation,
the corresponding network trace $\ntr$ is correct with respect to $N$.
\end{theorem}
\begin{proof}
The proof is by induction over the length $m$ of the execution.  In
the induction step, we show that (1) the {\sc Switch} rule can only
produce consistent event-sets (this follows directly from Lemma
\ref{lem:consistent}), and (2) when the {\sc In} rule tags a packet
$\pkt$ based on the local event-set $E$, that $E$ consists of exactly
the events that happened before $\pkt$ arrived (as ordered by the
happens-before relation).
\end{proof}

\renewcommand*{\proofname}{Proof}%

\section{Implementation and Evaluation}
\label{sec:eval}

We built a full-featured prototype implementation in OCaml.
\begin{compactitem}
\item We implemented the compiler described in Section \ref{sec:model}. This tool
accepts a Stateful NetKAT program, and produces the corresponding NES, with a standard
NetKAT program representing the configuration at each node.
We interface with Frenetic's NetKAT compiler to produce flow-table rules for each of these NetKAT programs.
\item We modified the OpenFlow 1.0 reference implementation %
to support the custom switch/controller needed to realize the runtime described
in Section \ref{sec:implementing}.
\item We built tools to automatically generate custom Mininet
scripts
to bring up the programmer-specified network topology, using switches/controller running
the compiled NES. We can then realistically simulate the whole system
using real network traffic.
\end{compactitem}

\paragraph{Research Questions.} To evaluate our approach, we wanted to obtain answers to the following questions.
\begin{compactenum}
\item How useful is our approach? Does it allow programmers to easily write
real-world network programs, and get the behavior they want?
\item What is the performance of our tools (compiler, etc.)?
\item How much does our correctness guarantee help? For instance, how do the
running network programs compare with \propose{uncoordinated} event-driven strategies?
\item How efficient are the implementations generated by our approach? For
instance, what about message overhead? State-change convergence time? Number of rules used?
\end{compactenum}

\noindent
We address \#1-3 through case studies on
real-world programming examples, and \#4 through quantitative performance measurements on simple
automatically-generated programs.
For the experiments, we assume that the programmer has first confirmed that the
program satisfies the conditions allowing proper compilation to an NES,
and we assume that the ETS has no loops.
Our tool could be modified to perform these checks
\xmnote{\FiveStar}{Q15}%
via basic algorithms operating on the ETS, but they
have not yet been implemented in the current prototype
(as mentioned in Section \ref{subsec:ets}, developing efficient algorithms
for these checks is left for future work).
Our experimental platform was an Ubuntu machine with 20GB
RAM and a quad-core Intel i5-4570 CPU (3.2 GHz).

To choose a representative set of realistic examples,
\xmnote{\FiveStar}{Q16}%
we first studied the examples addressed in other recent stateful network
programming approaches, such as
SNAP \cite{arashloo2015snap},
FlowLog \cite{nelson2014tierless},
Kinetic \cite{kim2015kinetic},
NetEgg \cite{yuan2015netegg}, and
FAST~\cite{moshref2014flow}, and
categorized them into three main groups:

\begin{compactitem}
\item {\em Protocols/Security}:
accessing streaming media across subnets, %
ARP proxy, %
{\bf firewall with authentication}, %
FTP monitoring, %
{\bf MAC learning}, %
{\bf stateful firewall}, %
TCP reassembly, %
Virtual Machine (VM) provisioning.
\item {\em Measurement/Performance}:
heavy hitter detection, %
{\bf bandwidth cap management (uCap)}, %
connection affinity in load balancing, %
counting domains sharing the same IP address, %
counting IP addresses under the same domain, %
elephant flows detection, %
link failure recovery, %
load balancing, %
network information base (NIB), %
QoS in multimedia streaming, %
rate limiting, %
sampling based on flow size, %
Snort flowbits, %
super spreader detection, %
tracking flow-size distributions.
\item {\em Monitoring/Filtering}:
application detection, %
DNS amplification mitigation, %
DNS TTL change tracking, %
DNS tunnel detection, %
{\bf intrusion detection system (IDS)}, %
optimistic ACK attack detection, %
phishing/spam detection, %
selective packet dropping, %
sidejack attack detection, %
stolen laptop detection, %
SYN flood detection, %
UDP flood mitigation, %
walled garden.
\end{compactitem}

\noindent
As we will see in the following section, our current prototype system is best suited for writing
programs such as the ones in the {\em Protocols/Security} category, since
some of the {\em Measurement/Performance}
programs require timers and/or integer counters,
and some of the {\em Monitoring/Filtering} programs require complex pattern matching of
(and table lookups based on)
sequences of packets---functionality which we do not (yet) natively support,
Thus, we have selected three examples from the first category, and one from each of the
latter two, corresponding to the boldface applications in the list.
We believe that these applications are representative of the basic types of behaviors seen in
the other listed applications.

\subsection{Case Studies}
\label{subsec:case}

\begin{figure*}[t]
\centering
\begin{minipage}{.33\textwidth}
\centerline{\includegraphics[trim = 0.0in 0.0in 0.0in 0.0in, clip,height=0.80in]{stateful}}
\centerline{(a, d)}
\end{minipage}\vrule\begin{minipage}{.33\textwidth}
\centerline{\includegraphics[trim = 0.0in 0.0in 0.0in 0.0in, clip,height=.80in]{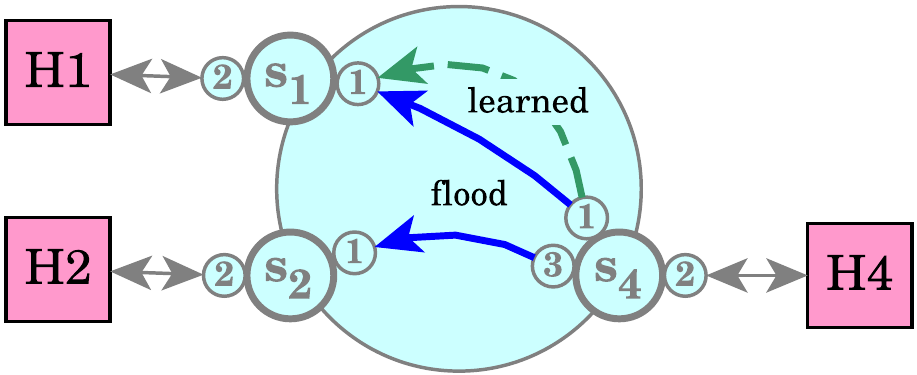}}
\centerline{(b)}
\end{minipage}\vrule\begin{minipage}{.33\textwidth}
\centerline{\includegraphics[trim = 0.0in 0.0in 0.0in 0.0in, clip,height=.80in]{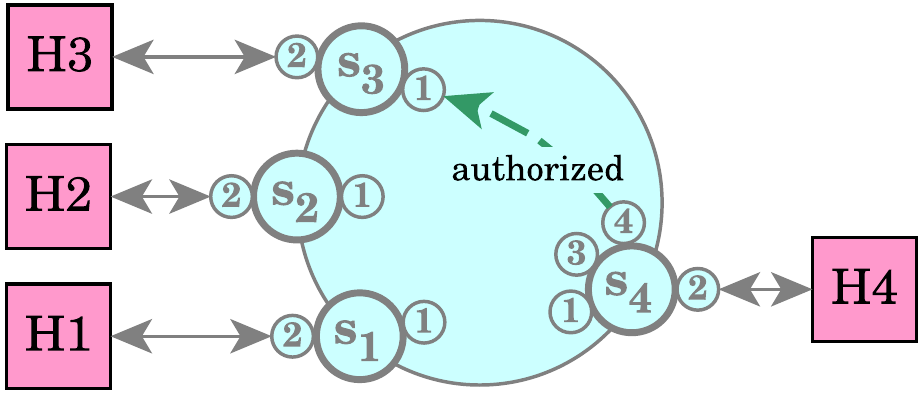}}
\centerline{(c, e)}
\end{minipage}
\caption{Topologies: (a) Firewall, (b) Learning Switch, (c) Authentication, (d) Bandwidth Cap, (e) Intrusion Detection System.}
\label{fig:examples-topo}
\end{figure*}

\newcommand{\thehrule}{\vspace{6pt}\hrule\vspace{6pt}}%
\begin{figure}[t]
\begin{minipage}{.1\linewidth}
\vspace{-1.0cm}
(a)

\vspace{1.25cm}

(b)

\vspace{1.75cm}
(c)

\vspace{2.75cm}
(d)

\vspace{3.0cm}
(e)
\end{minipage}\begin{minipage}{.9\linewidth}
\begin{lstlisting}[mathescape]
pt=2 $\land$ ip_dst=H4; pt$\leftarrow$1; (state=[0]; $(1{:}1){\rightarrowtriangle}(4{:}1){\rightarrowtriangle}\langle$state$\leftarrow$[1]$\rangle$ + state$\not=$[0]; $(1{:}1){\rightarrowtriangle}(4{:}1)$); pt$\leftarrow$2
+ pt=2 $\land$ ip_dst=H1; state=[1]; pt$\leftarrow$1; $(4{:}1){\rightarrowtriangle}(1{:}1)$; pt$\leftarrow$2
\end{lstlisting}
\thehrule
\begin{lstlisting}[mathescape]
pt=2 $\land$ ip_dst=H1; (pt$\leftarrow$1; $(4{:}1){\rightarrowtriangle}(1{:}1)$ + state=[0]; pt$\leftarrow$3; $(4{:}3){\rightarrowtriangle}(2{:}1)$); pt$\leftarrow$2
+ pt=2 $\land$ ip_dst=H4; pt$\leftarrow$1; $(1{:}1){\rightarrowtriangle}(4{:}1){\rightarrowtriangle}\langle$state$\leftarrow$[1]$\rangle$; pt$\leftarrow$2
+ pt=2; pt$\leftarrow$1; $(2{:}1){\rightarrowtriangle}(4{:}3)$; pt$\leftarrow$2
\end{lstlisting}
\thehrule
\begin{lstlisting}[mathescape]
state=[0] $\land$ pt=2 $\land$ ip_dst=H1; pt$\leftarrow$1; $(4{:}1){\rightarrowtriangle}(1{:}1){\rightarrowtriangle}\langle$state$\leftarrow$[1]$\rangle$; pt$\leftarrow$2
+ state=[1] $\land$ pt=2 $\land$ ip_dst=H2; pt$\leftarrow$3; $(4{:}3){\rightarrowtriangle}(2{:}1){\rightarrowtriangle}\langle$state$\leftarrow$[2]$\rangle$; pt$\leftarrow$2
+ state=[2] $\land$ pt=2 $\land$ ip_dst=H3; pt$\leftarrow$4; $(4{:}4){\rightarrowtriangle}(3{:}1)$; pt$\leftarrow$2
+ pt=2; pt$\leftarrow$1; ($(1{:}1){\rightarrowtriangle}(4{:}1)$ + $(2{:}1){\rightarrowtriangle}(4{:}3)$ + $(3{:}1){\rightarrowtriangle}(4{:}4)$); pt$\leftarrow$2
\end{lstlisting}
\thehrule
\begin{lstlisting}[mathescape]
pt=2 $\land$ ip_dst=H4;
    pt$\leftarrow$1; (
        state=[0]; $(1{:}1){\rightarrowtriangle}(4{:}1){\rightarrowtriangle}\langle$state$\leftarrow$[1]$\rangle$
        + state=[1]; $(1{:}1){\rightarrowtriangle}(4{:}1){\rightarrowtriangle}\langle$state$\leftarrow$[2]$\rangle$
        + state=[2]; $(1{:}1){\rightarrowtriangle}(4{:}1){\rightarrowtriangle}\langle$state$\leftarrow$[3]$\rangle$
                    $\vdots$
        + state=[10]; $(1{:}1){\rightarrowtriangle}(4{:}1){\rightarrowtriangle}\langle$state$\leftarrow$[11]$\rangle$
        + state=[11]; $(1{:}1){\rightarrowtriangle}(4{:}1)$
    ); pt$\leftarrow$2
+ pt=2 $\land$ ip_dst=H1; state$\not=$[11]; pt$\leftarrow$1; $(4{:}1){\rightarrowtriangle}(1{:}1)$; pt$\leftarrow$2
\end{lstlisting}
\thehrule
\begin{lstlisting}[mathescape]
pt=2 $\land$ ip_dst=H1; pt$\leftarrow$1; (state=[0]; $(4{:}1){\rightarrowtriangle}(1{:}1){\rightarrowtriangle}\langle$state$\leftarrow$[1]$\rangle$ + state$\not=$[0]; $(4{:}1){\rightarrowtriangle}(1{:}1)$); pt$\leftarrow$2
+ pt=2 $\land$ ip_dst=H2; pt$\leftarrow$3; (state=[1]; $(4{:}3){\rightarrowtriangle}(2{:}1){\rightarrowtriangle}\langle$state$\leftarrow$[2]$\rangle$ + state$\not=$[1]; $(4{:}3){\rightarrowtriangle}(2{:}1)$); pt$\leftarrow$2
+ pt=2 $\land$ ip_dst=H3; pt$\leftarrow$4; state$\not=$[2]; $(4{:}4){\rightarrowtriangle}(3{:}1)$; pt$\leftarrow$2
+ pt=2; pt$\leftarrow$1; ($(1{:}1){\rightarrowtriangle}(4{:}1)$ + $(2{:}1){\rightarrowtriangle}(4{:}3)$ + $(3{:}1){\rightarrowtriangle}(4{:}4)$); pt$\leftarrow$2
\end{lstlisting}
\end{minipage}
\caption{Programs: (a) Firewall, (b) Learning Switch, (c) Authentication, (d) Bandwidth Cap, (e) Intrusion Detection System.}
\label{fig:examples-code}
\end{figure}

\propose{In the first set of experiments,
\xmnote{\FiveStar}{Q2}%
we compare {\em correct} behavior
(produced by our implementation strategy) with that of an
{\em uncoordinated} update strategy.
We simulate an uncoordinated strategy in the following way:
events are sent to the controller,
which pushes updates to the switches (in an unpredictable order)
after a few-seconds time delay.
  We believe this delay is reasonable
  because heavily using the controller and frequently updating 
  switches can lead to delays between operations of several seconds in practice
  (e.g. \cite{jin2014dynamic} reports up to 10s for a single switch
  update).}

To show that problems still arise for smaller delays,
in the firewall experiment described next, we varied the
time delay in the uncoordinated strategy between 0ms and 5000ms (in increments of 100ms),
running the experiment 10 times for each. We then plotted the total number of incorrectly-dropped
packets with respect to delay. The results are shown in Figure \ref{fig:firewall_correct_corr}.
Note that even with a very small delay, the uncoordinated strategy
still always drops at least one packet.

\paragraph{Stateful Firewall.} 
The example in Figures~\ref{fig:examples-topo}-\ref{fig:examples-code}(a)
is a simplified stateful firewall. It always allows ``outgoing"
traffic (from H1 to H4), but only allows ``incoming" traffic (from H4
to H1) after the outside network has been contacted, i.e. ``outgoing"
traffic has been forwarded to H4.

Program $p$ corresponds to configurations
$C_{[0]} = \llbracket{p}\rrbracket_{[0]}$
and
$C_{[1]} = \llbracket{p}\rrbracket_{[1]}$.
In the former, only outgoing traffic is allowed, and in the
latter, both outgoing and incoming are allowed.
The ETS has the form $\{\langle [0] \rangle \xrightarrow{(dst{=}H4,\,{4{:}1})} \langle [1] \rangle\}$.
The NES has the form $\{E_0{=}\emptyset \rightarrow E_1{=}\{(dst{=}H4,\,{4{:}1})\}\}$,
where the $g$ is given by $g(E_0)=C_{[0]}$, $g(E_1)=C_{[1]}$.

The Stateful Firewall example took $0.013$s to compile, and produced a total of 18 flow-table rules.
In Figure \ref{fig:firewall_correct}(a), we show that the running firewall has
the expected behavior. We first try to ping H1 from H4 (the ``H4-H1"/red points),
which fails.
Then we ping H4 from H1 (the ``H1-H4"/orange points), which succeeds.
Again we try H4-H1, and now this succeeds, since the event-triggered state change occurred.

For the uncoordinated strategy, Figure \ref{fig:firewall_correct}(b) shows that
some of the H1-H4 pings get dropped (i.e. H1 does not hear back from H4), meaning the
state change did not behave as if it was caused immediately upon arrival of a packet at S4.

\begin{figure}[b]
\centering
\includegraphics[trim = 0.0in 0.0in 0.0in 0.0in, clip,width=0.875\linewidth]{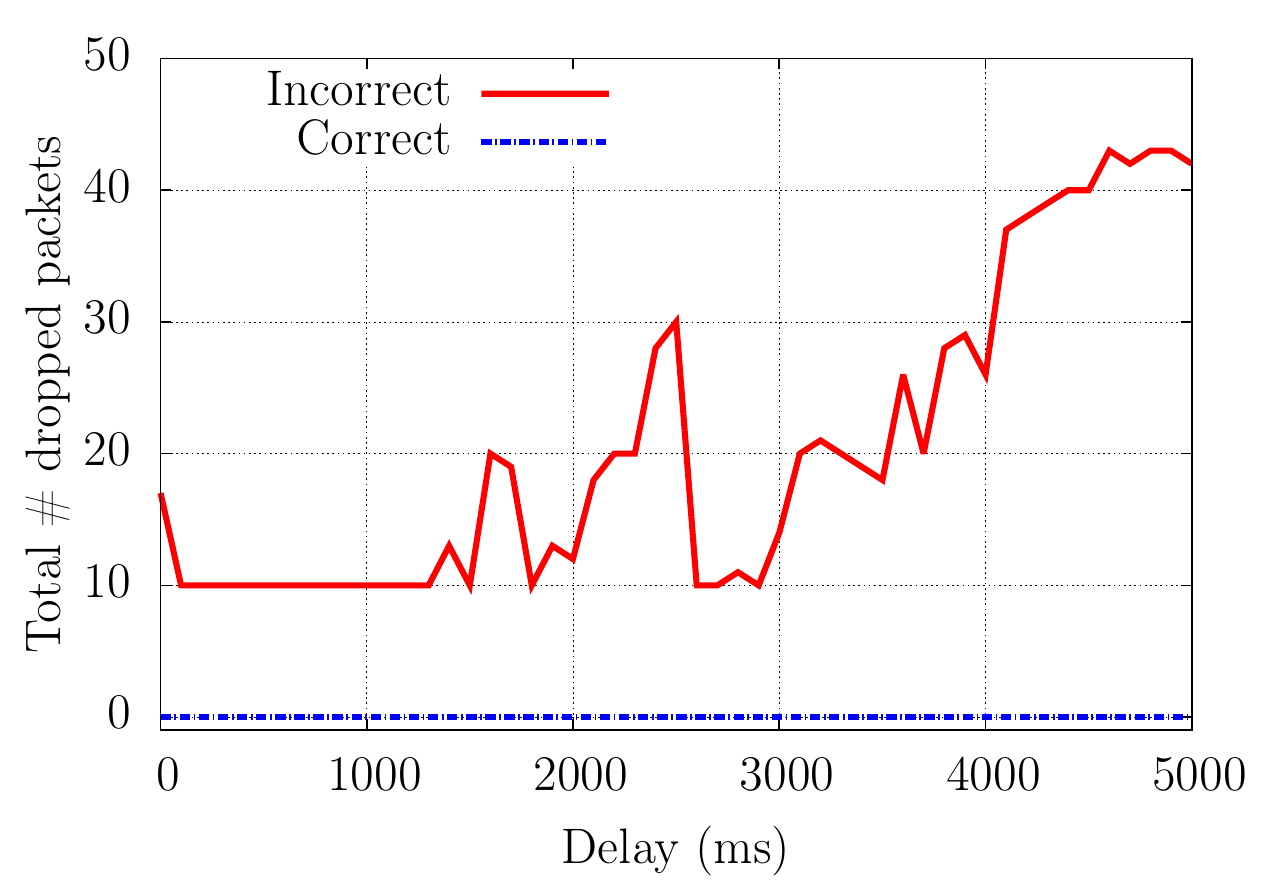}
\caption{Stateful Firewall: impact of delay.}
\label{fig:firewall_correct_corr}
\end{figure}

\begin{figure}[t]
\centering
\bgroup
\def\arraystretch{0.5}
\begin{tabular}{c}
\includegraphics[trim = 0.0in 0.35in 0.0in 0.0in, clip,width=0.86\linewidth]{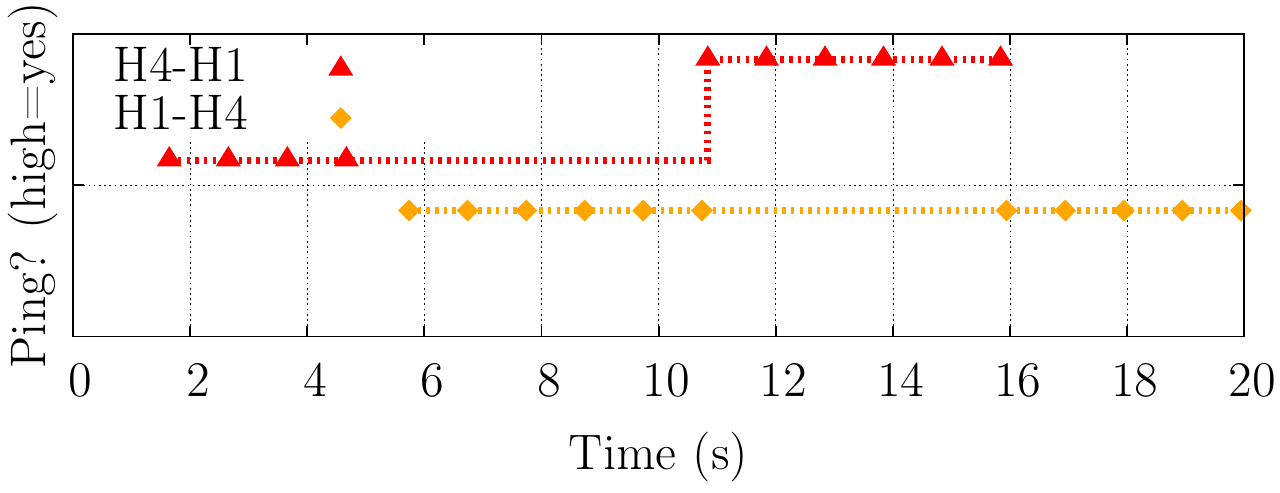} \\
{\scriptsize (a)} \\
\includegraphics[trim = 0.0in 0.0in 0.0in 0.0in, clip,width=0.86\linewidth]{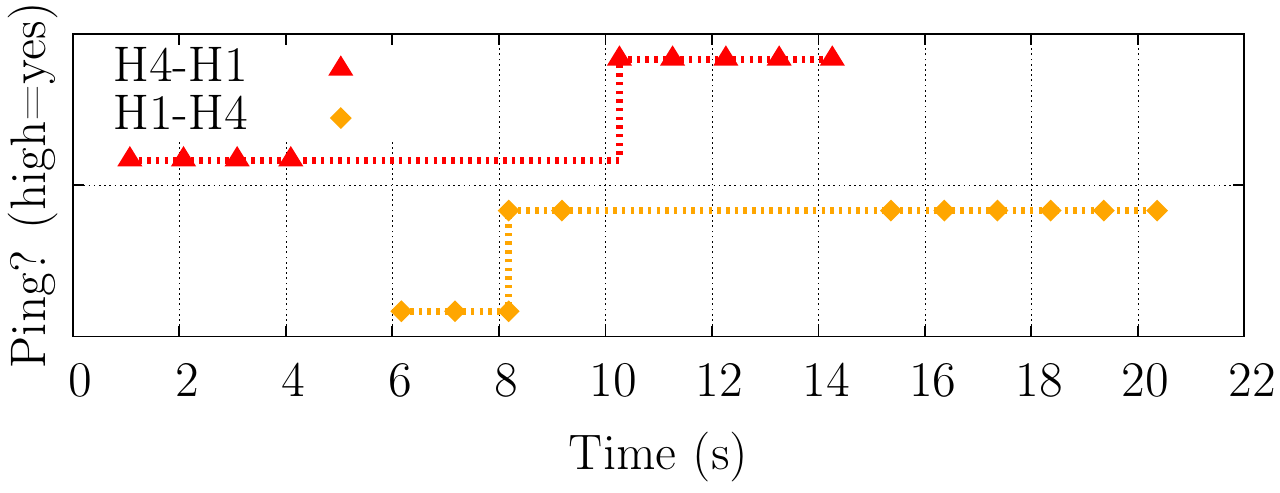} \\
{\scriptsize (b)}
\end{tabular}
\egroup
\caption{Stateful Firewall: (a) correct vs. (b) incorrect.}
\label{fig:firewall_correct}
\end{figure}

\paragraph{Learning Switch.} The example in Figures~\ref{fig:examples-topo}-\ref{fig:examples-code}(b)
is a simple learning switch. Traffic from H4 to H1 is flooded (sent
to both H1 and H2), until H4 receives a packet from H1, at which
point it ``learns" the address of H1,
and future traffic from H4 to H1 is sent only to H1.

\begin{figure}[b]
\centering
\bgroup
\def\arraystretch{0.5}
\begin{tabular}{c}
\includegraphics[trim = 0.0in 0.345in 0.0in 0.0in, clip,width=0.86\linewidth]{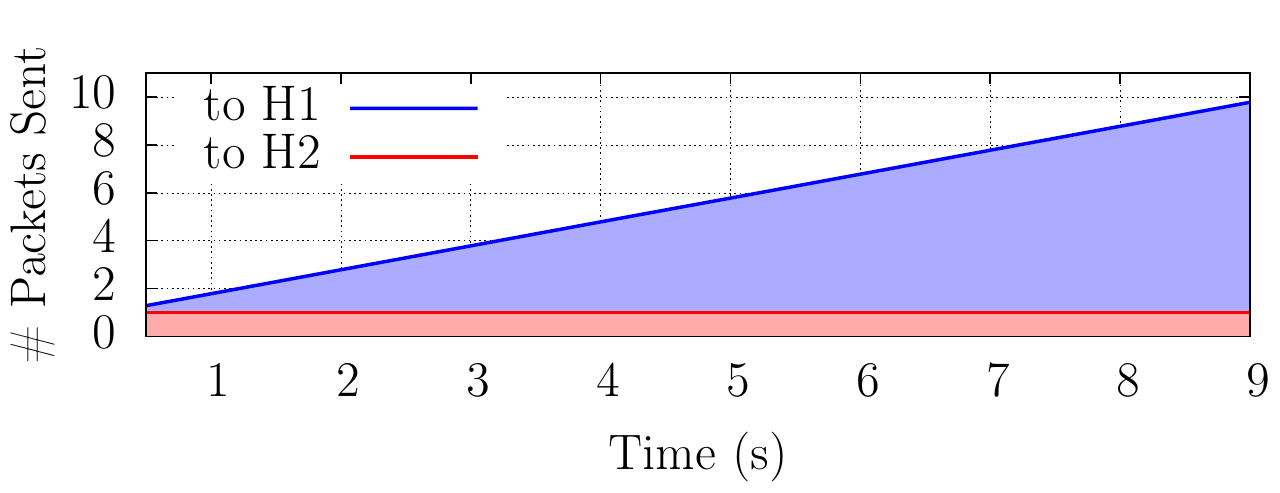} \\
{\scriptsize (a)} \\
\includegraphics[trim = 0.0in 0.0in 0.0in 0.0in, clip,width=0.86\linewidth]{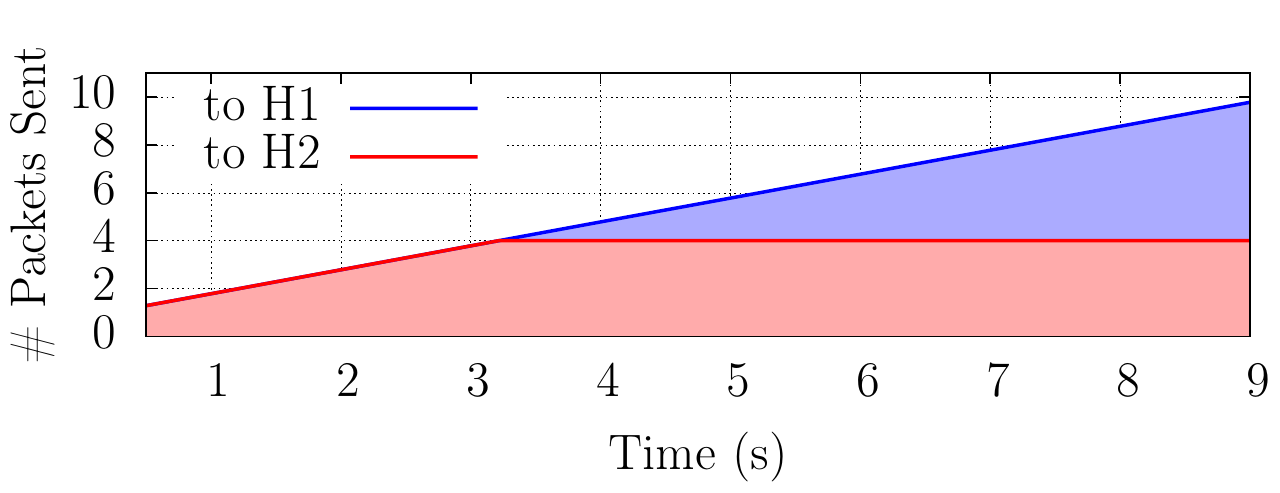} \\
{\scriptsize (b)}
\end{tabular}
\egroup
\caption{Learning Switch: (a) correct vs. (b) incorrect.}
\label{fig:learning_correct}
\end{figure}

This program $p$ corresponds to two configurations
$C_{[0]} = \llbracket{p}\rrbracket_{[0]}$
and
$C_{[1]} = \llbracket{p}\rrbracket_{[1]}$.
In the former, flooding occurs from H4, and in the
latter, packets from H4 are forwarded directly to H1.
The ETS has the form $\{\langle [0] \rangle \xrightarrow{(dst{=}H4,\,{4{:}1})} \langle [1] \rangle\}$.
The NES has the form $\{E_0{=}\emptyset \rightarrow E_1{=}\{(dst{=}H4,\,{4{:}1})\}\}$,
where the $g$ is given by $g(E_0)=C_{[0]}$, $g(E_1)=C_{[1]}$.

This only allows learning for a single host (H1), but
we could easily add learning for H2 by using a different index in the vector-valued {\em state} field:
we could replace $\kw{state}$ in Figure~\ref{fig:examples-code}(b) with $\kw{state}(0)$,
and union the program (using the NetKAT ``$+$" operator) with another instance of
Figure \ref{fig:examples-code}(b) which learns for H2 and uses $\kw{state}(1)$.

The Learning Switch example took $0.015$s to compile, and produced a total of 43 flow-table rules. We again compare the behavior of our correct implementation with that of an
implementation which uses an uncoordinated update strategy.
We first ping H1 from H4. Expected behavior is shown in Figure \ref{fig:learning_correct}(a), 
where the first packet is flooded to both H1 and H2, but then H4 hears a
reply from H1, causing the state change (i.e. learning H1's address),
and all subsequent packets are sent only to H1.
In Figure \ref{fig:learning_correct}(b), however, since the state change can be
delayed, multiple packets are sent to H2, even after H4 has seen a reply from H1.

\paragraph{Authentication.} In this example, shown in Figures~\ref{fig:examples-topo}-\ref{fig:examples-code}(c), the untrusted host H4 wishes to contact H3,
but can only do so after contacting H1 and then H2, in
that order.

\begin{figure}
\centering
\bgroup
\def\arraystretch{0.5}
\begin{tabular}{c}
\includegraphics[trim = 0.0in 0.35in 0.0in 0.0in, clip,width=0.86\linewidth]{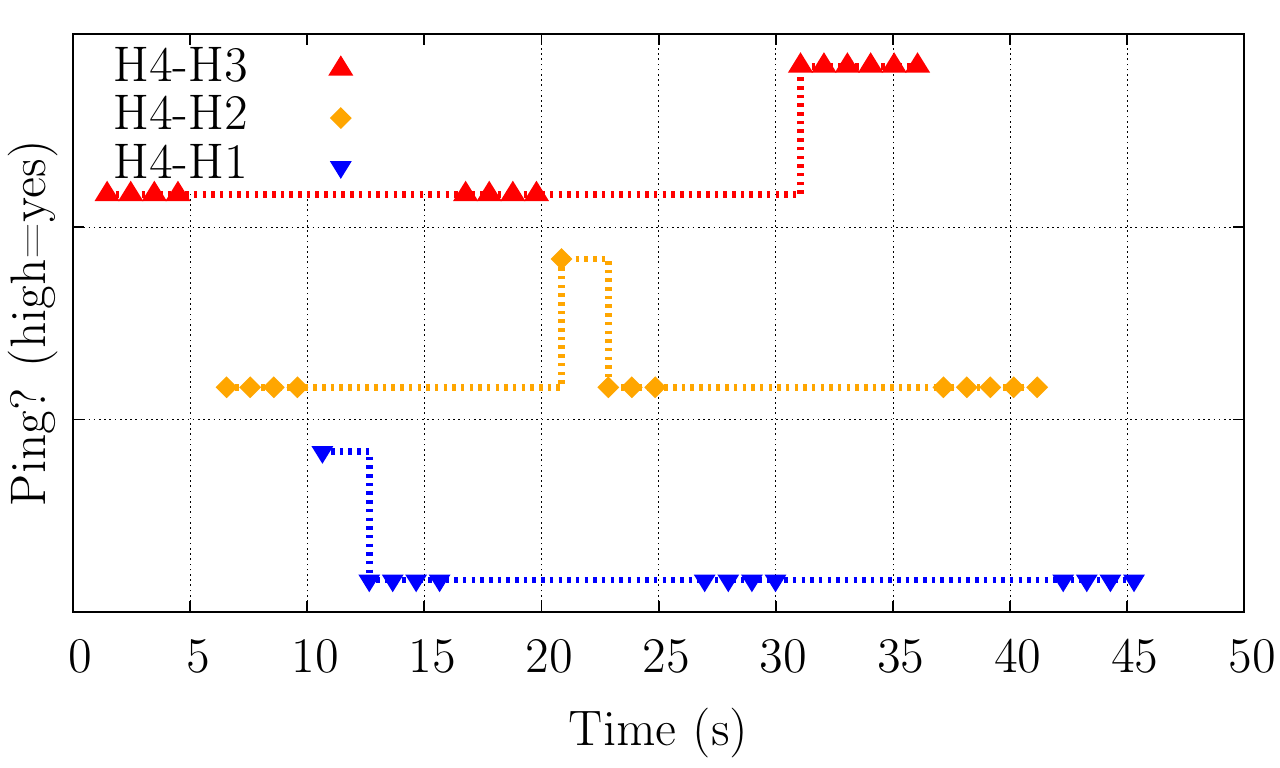} \\
{\scriptsize (a)} \\
\includegraphics[trim = 0.0in 0.0in 0.0in 0.0in, clip,width=0.86\linewidth]{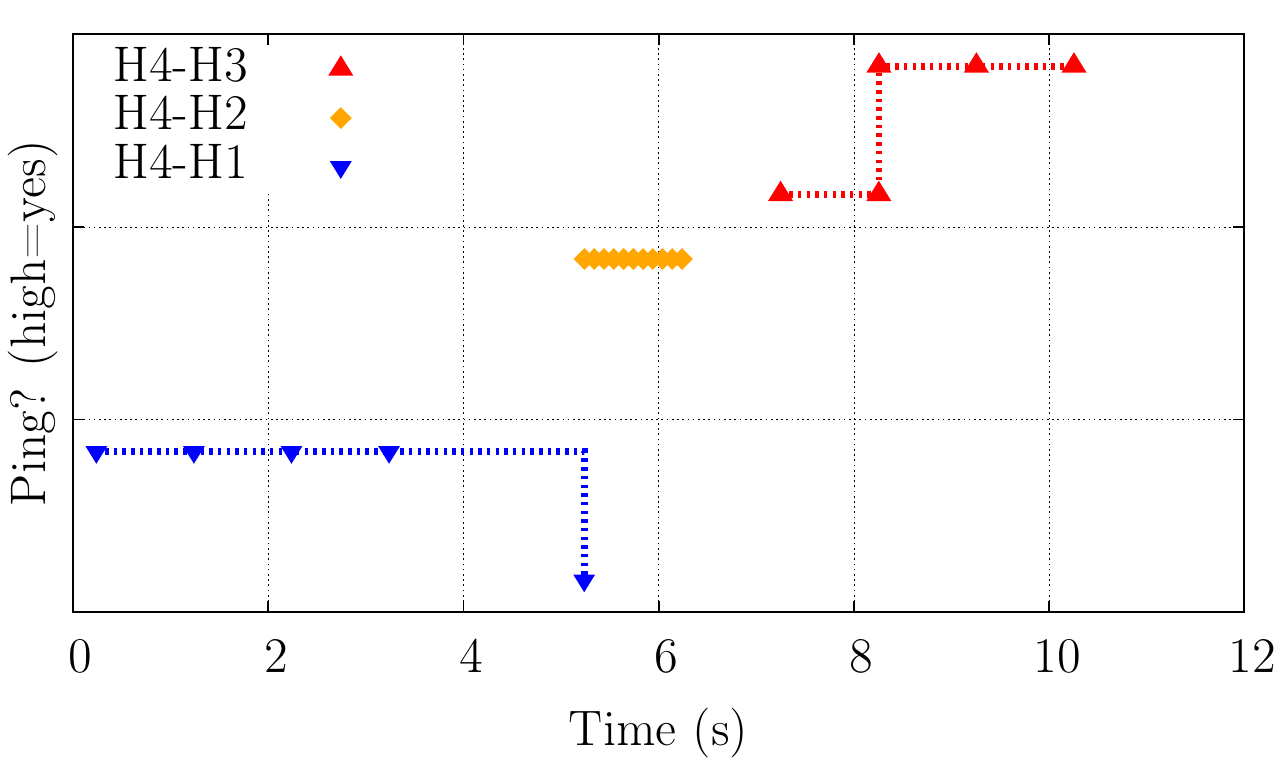} \\
{\scriptsize (b)}
\end{tabular}
\egroup
\caption{Authentication: (a) correct vs. (b) incorrect.}
\label{fig:knocking_correct}
\end{figure}

This program $p$ corresponds to three configurations:
$C_{[0]} = \llbracket{p}\rrbracket_{[0]}$ in which only H4-H1 traffic is enabled,
$C_{[1]} = \llbracket{p}\rrbracket_{[1]}$ in which only H4-H2 traffic is enabled,
and $C_{[2]} = \llbracket{p}\rrbracket_{[2]}$ which finally allows H4 to communicate with H3.
The ETS has the form $\{\langle [0] \rangle \xrightarrow{(dst{=}H1,\,{1{:}1})} \langle [1] \rangle \xrightarrow{(dst{=}H2,\,{2{:}1})} \langle [2] \rangle\}$.
The NES has the form $\{E_0{=}\emptyset \rightarrow E_1{=}\{(dst{=}H1,\,{1{:}1})\} \rightarrow E_2{=}\{(dst{=}H1,\,{1{:}1}),(dst{=}H2,\,\allowbreak{2{:}1})\}\}$,
where the $g$ function is given by $g(E_0)=C_{[0]}$, $g(E_1)=C_{[1]}$, $g(E_2)=C_{[2]}$.

The Authentication example took $0.017$s to compile, and produced a total of 72 flow-table rules.
In Figure \ref{fig:knocking_correct}(a) we demonstrate the correct behavior of the
program, by first trying (and failing) to ping H3 and H2 from H4, then successfully
pinging H1, again failing to ping H3 (and H1), and finally succeeding in pinging H3.
The incorrect (uncoordinated) implementation in Figure \ref{fig:knocking_correct}(b) allows an incorrect behavior where
we can successfully ping H1 and then H2, but then fail to ping H3 (at least temporarily).

\paragraph{Bandwidth Cap.}
The Figure~\ref{fig:examples-topo}-\ref{fig:examples-code}(d) example
is a simplified bandwidth cap implementation. It allows ``outgoing"
traffic (H1-H4), but only until the limit of $n$ packets
has been reached,
at which point the service provider replies with a notification message, and
disallows the ``incoming" path.
In this experiment, we use a bandwidth cap of $n=10$ packets.

\begin{figure}[t]
\centering
\bgroup
\def\arraystretch{0.5}
\begin{tabular}{c}
\includegraphics[trim = 0.0in 0.35in 0.0in 0.0in, clip,width=0.86\linewidth]{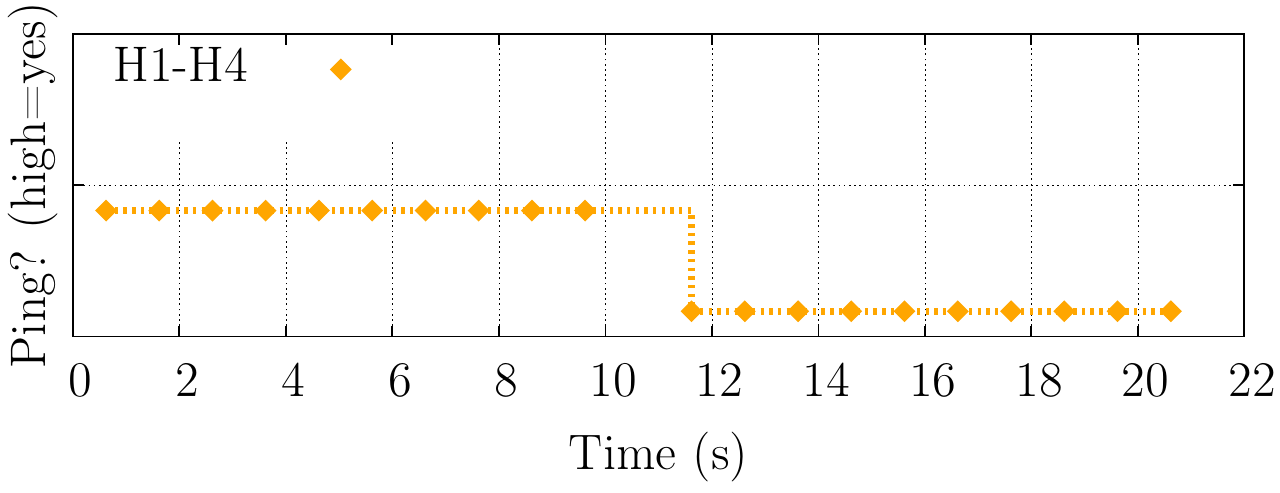} \\
{\scriptsize (a)} \\
\includegraphics[trim = 0.0in 0.0in 0.0in 0.0in, clip,width=0.86\linewidth]{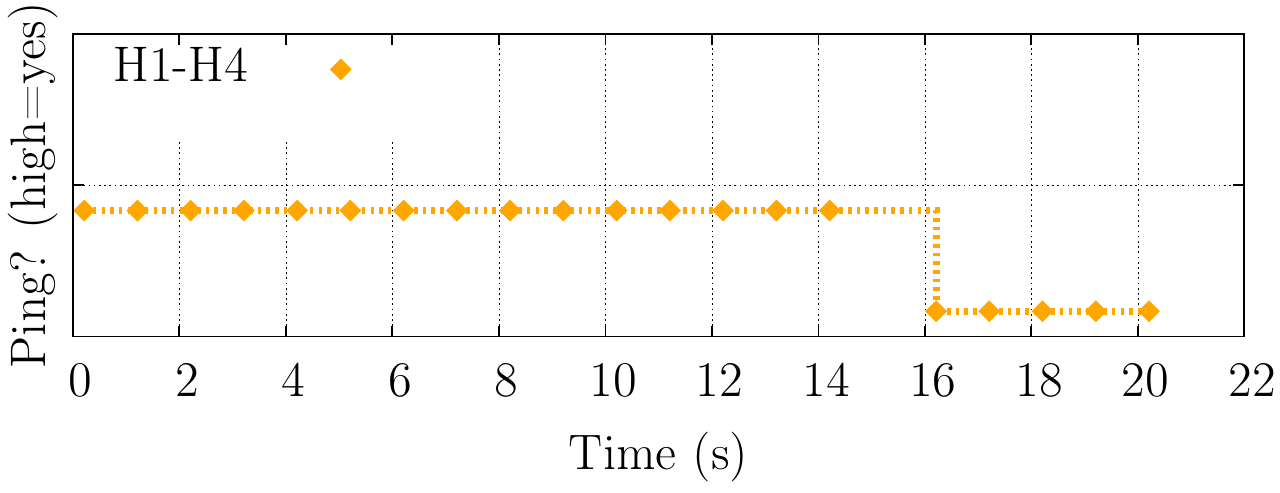} \\
{\scriptsize (b)}
\end{tabular}
\egroup
\caption{Bandwidth Cap: (a) correct vs. (b) incorrect.}
\label{fig:ucap_correct}
\end{figure}

Program $p$ corresponds to configurations
$C_{[0]}{=}\llbracket{p}\rrbracket_{[0]},\allowbreak \cdots,\allowbreak C_{[n]}{=}\llbracket{p}\rrbracket_{[n]}$,
which all allow incoming/outgoing traffic, and a configuration 
$C_{[n+1]}{=}\llbracket{p}\rrbracket_{[n+1]}$ which disallows the incoming traffic.
The ETS has the form $\{\langle [0] \rangle \xrightarrow{(dst{=}H4,\,{4{:}1})} \langle [1] \rangle \xrightarrow{(dst{=}H4,\,{4{:}1})} \cdots \xrightarrow{(dst{=}H4,\,{4{:}1})} \langle [n+1] \rangle\}$.
The NES has the form $\{E_0{=}\emptyset \rightarrow E_1{=}\{(dst{=}H4,\,{4{:}1})\} \rightarrow \cdots
\rightarrow E_{n+1}{=}\{(dst{=}H4,\,{4{:}1})_0,\cdots,(dst{=}H4,\,{4{:}1})_n\} \}$,
where the $g$ is given by $g(E_0)=C_{[0]}, \cdots, g(E_{n+1})=C_{[n+1]}$.
Note that the subscripts on events in the NES event-sets (e.g. the ones in $E_{n+1}$) indicate ``renamed" copies of
the same event (as described in Section \ref{subsec:ets}).

The Bandwidth Cap example took $0.023$s to compile, and produced a total of $158$ flow-table rules.
In Figure \ref{fig:ucap_correct}(a), we show that the running example has
the expected behavior. We send pings from H1 to H4, of which exactly 10 succeed, meaning we have
reached the bandwidth cap.
Using the uncoordinated update strategy in Figure \ref{fig:ucap_correct}(b), we again send pings
from H1 to H4, but in this case, 15 are successful, exceeding the bandwidth cap.

\paragraph{Intrusion Detection System.} In this example, shown in Figures~\ref{fig:examples-topo}-\ref{fig:examples-code}(e),
the external host H4 is initially free to communicate with the internal hosts H1, H2, and H3.
However, if H4 begins engaging in some type of suspicious activity (in this case,
beginning to scan through the hosts, e.g. contacting H1 and then H2, in that order),
the activity is thwarted (in this case, by cutting off access to H3).

\begin{figure}
\centering
\bgroup
\def\arraystretch{0.5}
\begin{tabular}{c}
\includegraphics[trim = 0.0in 0.35in 0.0in 0.0in, clip,width=0.86\linewidth]{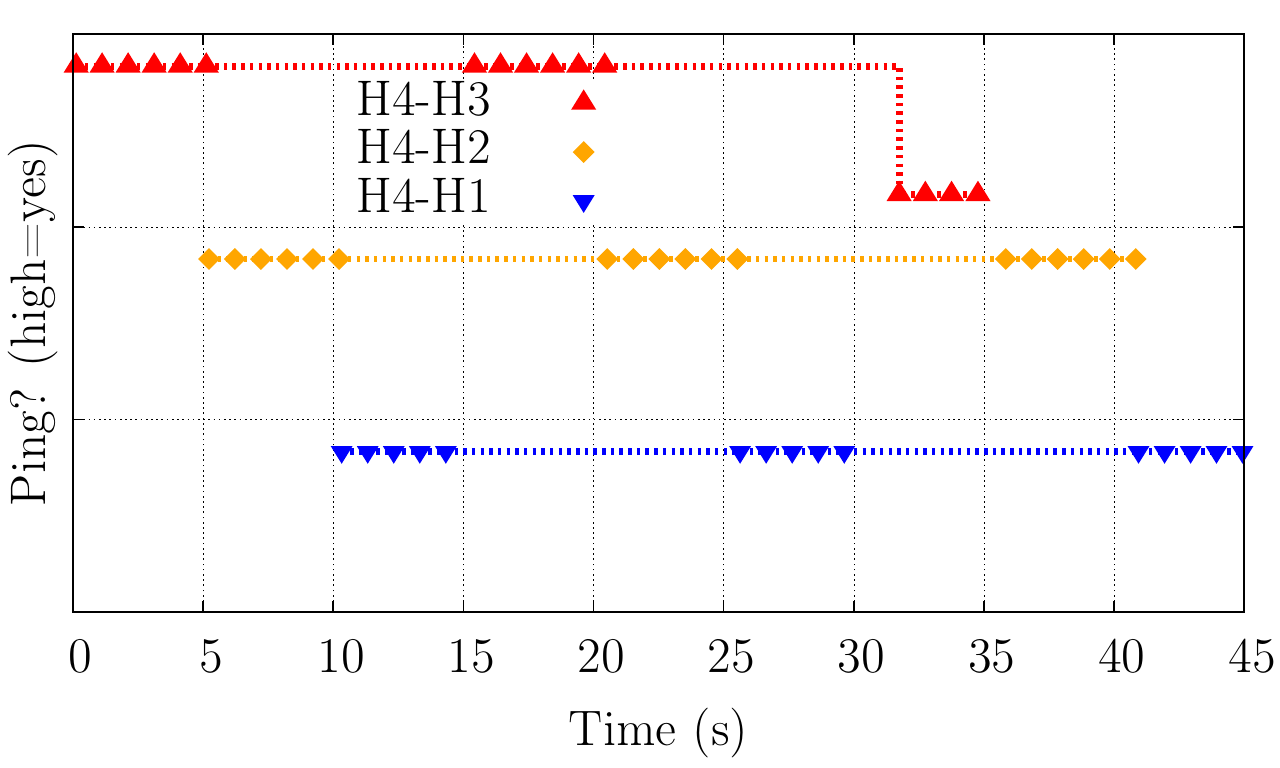} \\
{\scriptsize (a)} \\
\includegraphics[trim = 0.0in 0.0in 0.0in 0.0in, clip,width=0.86\linewidth]{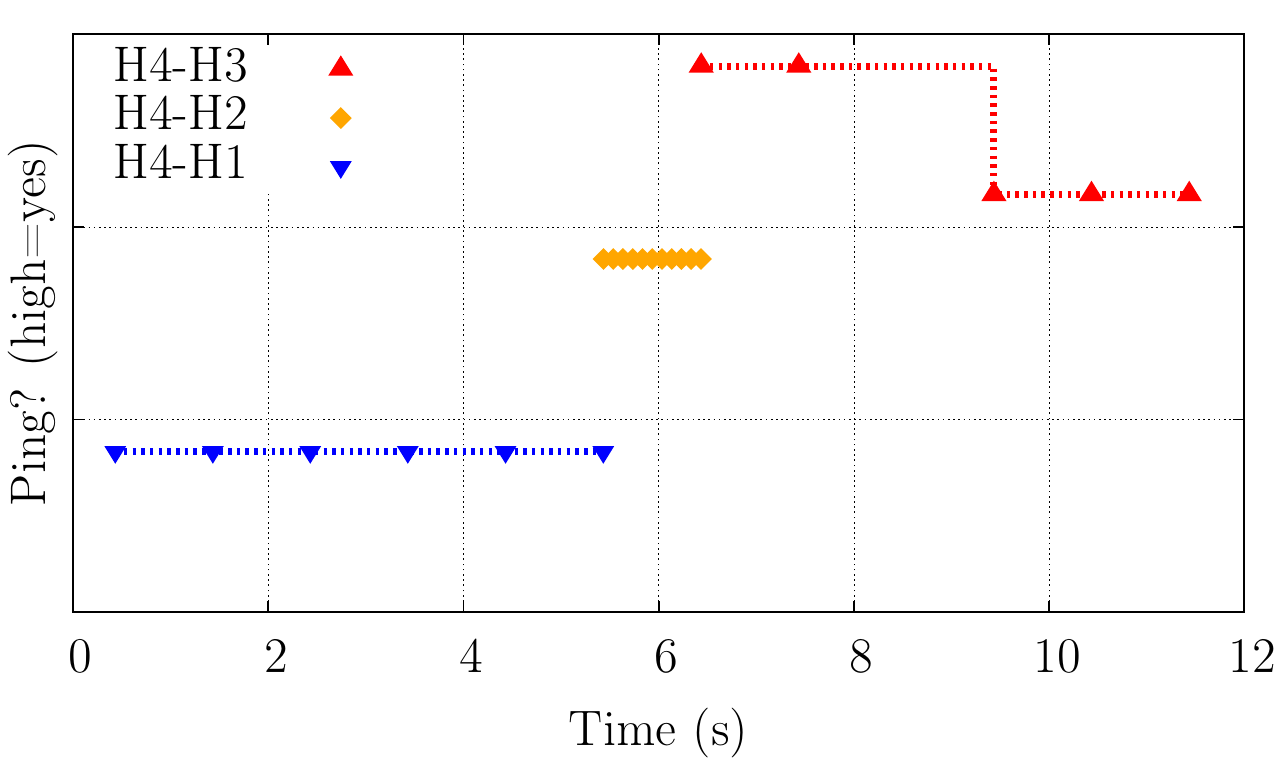} \\
{\scriptsize (b)}
\end{tabular}
\egroup
\caption{Intrusion Detection System: (a) correct vs. (b) incorrect.}
\label{fig:ids_correct}
\end{figure}

This program $p$ corresponds to three configurations:
$C_{[0]} = \llbracket{p}\rrbracket_{[0]}$ and
$C_{[1]} = \llbracket{p}\rrbracket_{[1]}$, in which all traffic is enabled,
and $C_{[2]} = \llbracket{p}\rrbracket_{[2]}$ in which H4-H3 communication is disabled.
The ETS has the form $\{\langle [0] \rangle \xrightarrow{(dst{=}H1,\,{1{:}1})} \langle [1] \rangle \xrightarrow{(dst{=}H2,\,{2{:}1})} \langle [2] \rangle\}$.
The NES has the form $\{E_0{=}\emptyset \rightarrow E_1{=}\{(dst{=}H1,\,{1{:}1})\} \rightarrow E_2{=}\{(dst{=}H1,\,{1{:}1}),(dst{=}H2,\,\allowbreak{2{:}1})\}\}$,
where the $g$ function is given by $g(E_0)=C_{[0]}$, $g(E_1)=C_{[1]}$, $g(E_2)=C_{[2]}$.

This IDS example took $0.021$s to compile and produced $152$ flow-table rules.
In Figure \ref{fig:ids_correct}(a), we demonstrate the correct behavior of the
program, by first successfully pinging H3, H2, H1, H3, H2, H1 (in that order) from H4.
This results in a situation where we have contacted H1 and then H2, causing the third attempt to
contact H3 to be blocked (H4-H3 pings dropped).
The incorrect (uncoordinated) implementation in Figure \ref{fig:ids_correct}(b) allows a faulty
behavior where we can successfully ping H1 and then H2 (in that order), but subsequent H4-H3 traffic is
still enabled temporarily.

\subsection{Quantitative Results}

In this experiment, we automatically generated some event-driven programs which specify
that two hosts H1 and H2 are connected to opposite sides of a ring of switches. Initially,
traffic is forwarded clockwise, but when a specific switch detects a (packet) event, the
configuration changes to forward counterclockwise.
We increased the ``diameter" of the ring (distance from H1 to H2) up to 8, as shown in
Figure \ref{fig:bandwidth}, and performed the following two experiments.
\begin{compactenum}
\item We used {\tt iperf} to measure H1-H2 TCP/UDP bandwidth, and compared the
performance of our running event-driven program, versus that of the initial (static)
configuration of the program running on un-modified OpenFlow 1.0 reference
switches/controller. Figure \ref{fig:bandwidth}(a) shows that our performance
(solid line) is very close to the performance of a system which does not do packet tagging, event
detection, etc. (dashed line)---we see around 6\% performance degradation on average (note
that the solid and dashed lines almost coincide).
\item We measured maximum and average time needed for a switch to learn about the event.
The ``Max." and ``Avg." bars in Figure \ref{fig:bandwidth}(b) are these numbers when the
controller does not assist in disseminating events (i.e. only the packet digest is used),
and the other columns are the maximum and average when the controller does so.
\end{compactenum}

\begin{figure}
\centering
\bgroup
\def\arraystretch{0.5}
\begin{tabular}{c}
\includegraphics[trim = 0.0in 0.25in 0.0in 0.0in, clip,width=0.86\linewidth]{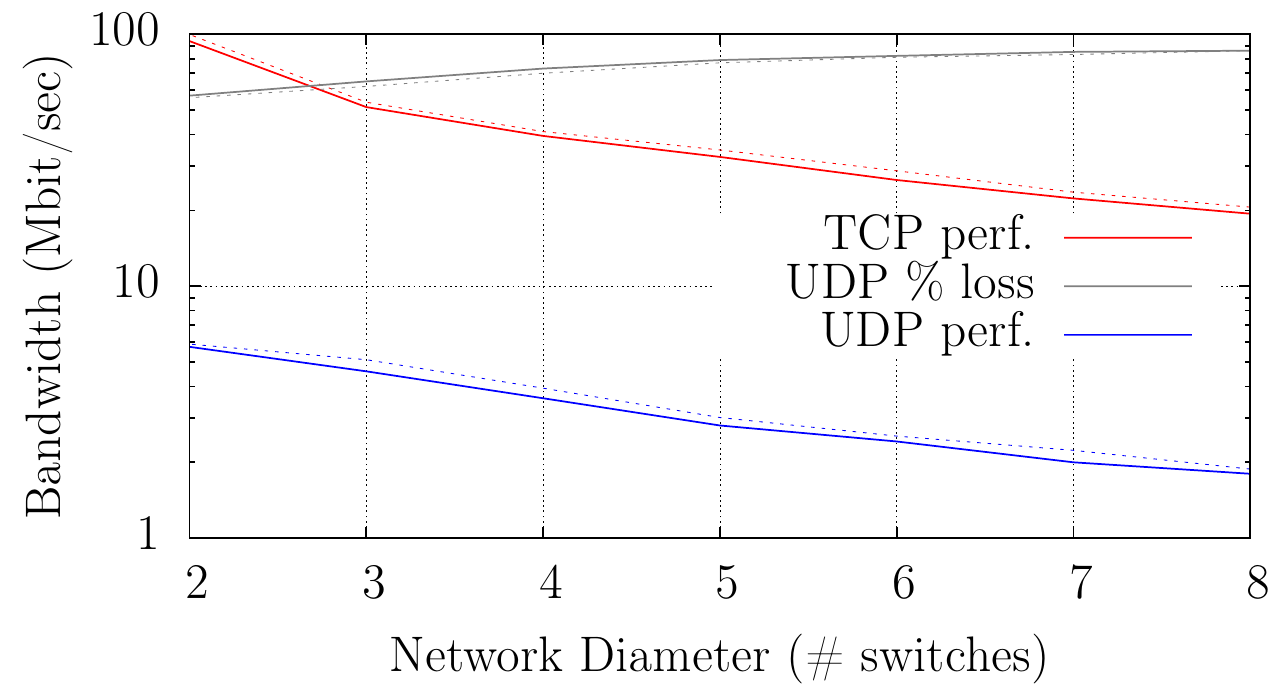} \\
{\scriptsize (a)} \\
\includegraphics[trim = 0.0in 0.0in 0.0in 0.0in, clip,width=0.86\linewidth]{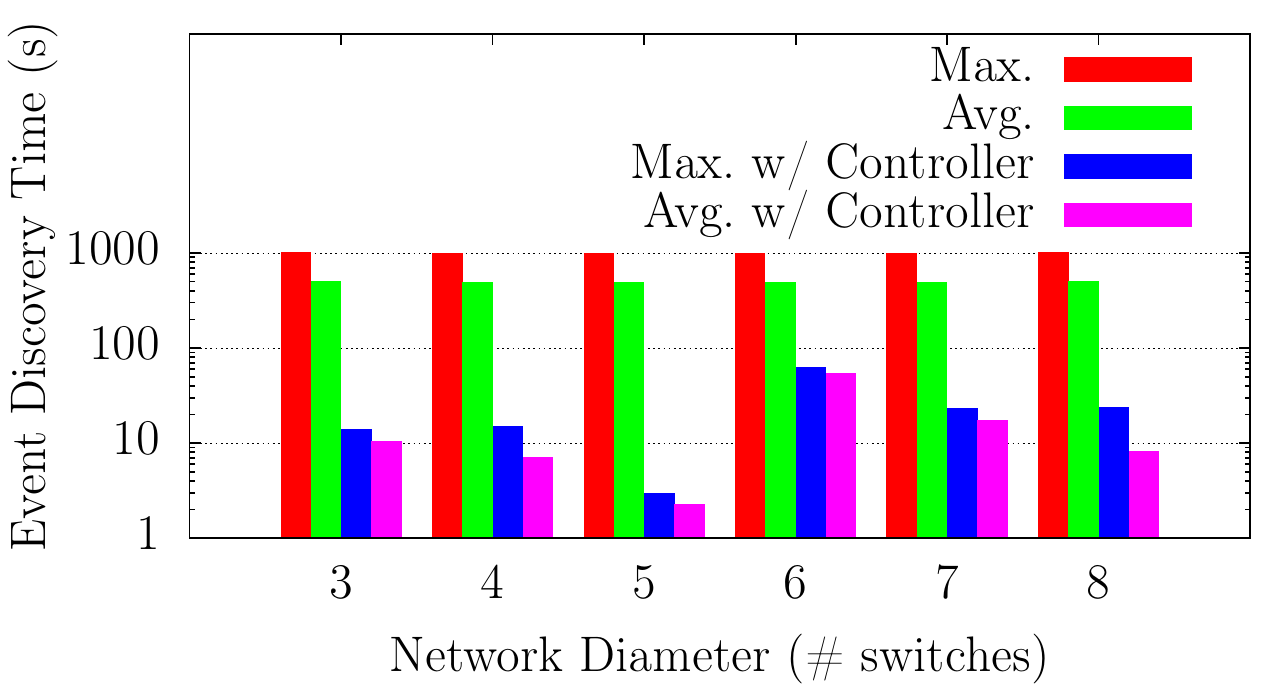} \\
{\scriptsize (b)}
\end{tabular}
\egroup
\caption{Circular Example: (a) bandwidth (solid line is ours, dotted line is reference implementation)
and (b) convergence.}
\label{fig:bandwidth}
\end{figure}

\begin{figure}[t]
\footnotesize
\begin{tikzpicture}
  \begin{axis}[%
    xlabel style={align=center,text width=\linewidth}, xlabel=Number of Rules w/ Heuristic,%
    ylabel=\# Original Rules,%
    log ticks with fixed point,
    legend pos = north west,%
    width=\linewidth,
    height=0.4\linewidth,
    legend style={fill=none},
    legend cell align=left,
    grid=major,%
        ]
    \addplot [mark=x,only marks, red, mark size=2pt] table [%
    x=good,%
    y=bad,%
    col sep=comma,%
    ignore chars=?] 
    {rule-num.csv};
\addplot[color=darkgray,very thick,dotted] coordinates {
            (275, 275)
            (345, 345)
        }  node[pos=0.9,pin={[pin distance=4pt]95:{\color{gray}$x=y$}}] {};
  \end{axis}
\end{tikzpicture}
\caption{Heuristic: reducing the number of rules.}
\label{fig:heuristic}
\end{figure}
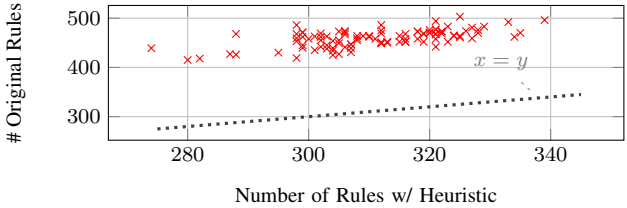

\subsection{Optimizations}
\label{sec:opt}

When a configuration change occurs,
the old and new configurations are often similar, 
differing only in a subset of flow-table rules.
Tables are commonly stored in TCAM memory on switches, which is
limited/costly, so it is undesirable to store duplicate rules. 
As mentioned in Section \ref{subsec:building_blocks}, each
of our rules is guarded by its configuration's numeric ID.
If the same rule occurs in several configurations having
IDs with the same (binary) high-order bits,
intuitively we can reduce space
usage by keeping a single copy of the rule, and guarding it
with a configuration ID having the shared high-order bits, and {\em wildcarded} low-order bits.
For example, if rule $r$ is used in two different 
configurations having IDs $2$ (binary $10$) and $3$ (binary $11$), 
we can wildcard the lowest bit $(1*)$, and keep a single rule $(1*)r$ having this wildcarded
guard, instead of two copies of $r$, with the ``$10$" and ``$11$" guards.
Ideally, we would like to (re)assign numeric IDs to the configurations, such
that maximal sharing of this form is achieved.

We formalize the problem as follows. 
Assume there is a set of all possible rules ${\mathcal R}$.
A configuration $C$ is a subset of these rules $C\subseteq {\mathcal R}$.
Assume there are $k$ bits in a configuration ID.
Without loss of generality we assume there are exactly $2^k$ configurations 
(if there are fewer, we can add dummy configurations, each containing all rules in ${\mathcal R}$). 
For a given set of configurations, we construct a \emph{trie} having all
of the configurations at the leaves.
This trie is a complete binary tree in which every node is marked 
with (1) a wildcarded mask that represents the configuration IDs of its children,
and (2) the intersection of the rule-sets of its children.

Consider configurations 
$C_0=\{r_1,r_2\}$,
$C_1=\{r_1,r_3\}$,
$C_2=\{r_2,r_3\}$,
$C_3=\{r_1,r_2\}$.
Figure~\ref{fig:orders} shows two different assignments of configurations to the leaves
of tries.
The number of rules for trie (a) is 6: $(0*)r_1$ , $(00)r_2$ , $(01)r_3$ , $(1*)r_2$, $(10)r_3$, $(11)r_1$.
The number of rules for trie (b) is 5: $(0*)r_1$ , $(0*)r_2$ , $(1*)r_3$ , $(10)r_1$, $(11)r_2$.
Intuitively, this is because the trie (b) has larger sets in the interior. 
Our polynomial heuristic follows that basic intuition:
it constructs the trie from the leaves up, at each level pairing nodes in a way that
maximizes the sum of the cardinalities of their sets.
This does not always produce the global maximum rule sharing, but we find that it
produces good results in practice.

\begin{figure}
{\footnotesize\hspace{-0.5cm}
\begin{tabular}{ c@{\hskip -0.3cm} c@{\hskip 0.1cm} }
\begin{tikzpicture}[level distance=0.8cm,
  level 1/.style={sibling distance=2cm},
  level 2/.style={sibling distance=1cm}]
  \node {\begin{tabular}{c} $**$ , $\emptyset$ \end{tabular}}
    child {node {\begin{tabular}{c} $0*$ , $\{r_1\}$ \end{tabular}}
      child {node {\begin{tabular}{c} $00$, \\ {\footnotesize $\{r_1,r_2\}$} \end{tabular}}}
      child {node {\begin{tabular}{c} $01$, \\ {\footnotesize $\{r_1,r_3\}$} \end{tabular}}}
    }
    child {node {\begin{tabular}{c} $1*$ , $\{r_2\}$ \end{tabular}}
    child {node {\begin{tabular}{c} $10$, \\ {\footnotesize $\{r_2,r_3\}$} \end{tabular}}}
      child {node {\begin{tabular}{c} $11$, \\ {\footnotesize $\{r_1,r_2\}$} \end{tabular}}}
    };
\end{tikzpicture}
&
\begin{tikzpicture}[level distance=0.8cm,
  level 1/.style={sibling distance=2cm},
  level 2/.style={sibling distance=1cm}]
  \node {\begin{tabular}{c} $**$ , $\emptyset$ \end{tabular}}
    child {node {\begin{tabular}{c} $0*$ , $\{r_1,r_2\}$ \end{tabular}}
      child {node {\begin{tabular}{c} $00$, \\ {\footnotesize $\{r_1,r_2\}$} \end{tabular}}}
      child {node {\begin{tabular}{c} $01$, \\ {\footnotesize $\{r_1,r_2\}$} \end{tabular}}}
    }
    child {node {\begin{tabular}{c} $1*$ , $\{r_3\}$ \end{tabular}}
    child {node {\begin{tabular}{c} $10$, \\ {\footnotesize $\{r_1,r_3\}$} \end{tabular}}}
      child {node {\begin{tabular}{c} $11$, \\ {\footnotesize $\{r_2,r_3\}$} \end{tabular}}}
    };
\end{tikzpicture}
\\
{\scriptsize (a)} & {\scriptsize (b)}
\end{tabular}
\caption{Heuristic: two different tries for the same configurations.}
\label{fig:orders}
}\end{figure}
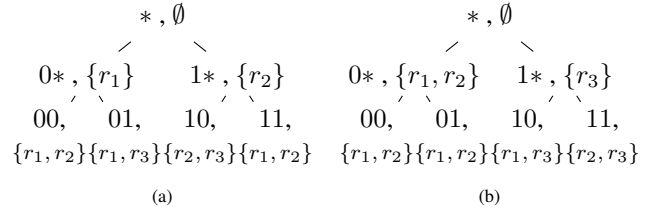

As indicated by the Figure~\ref{fig:heuristic} result (64 randomly-generate configurations w/ 20 rules),
on average, rule savings was about 32\% of the original number of rules.
We also ran this on the previously-discussed Firewall, Learning Switch, Authentication, Bandwidth Cap, and IDS examples,
and got rule reductions of $18 \rightarrow 16$, $43 \rightarrow 27$, $72 \rightarrow 46$,
$158 \rightarrow 101$, and $152 \rightarrow 133$ respectively.

\section{Related Work}
\label{sec:related}

\paragraph{Network Updates, Verification, and Synthesis.}
We already briefly mentioned an early approach known as
consistent updates \cite{reitblatt2012abstractions}. This work was
followed by update techniques that respect other correctness properties 
\cite{ludwig2014good}
\cite{jin2014dynamic}
\cite{godfrey2015nsdi}
\cite{mcclurg2015efficient}.
These approaches for expressing and verifying correctness of network updates
work in terms of {\em individual} packets.

In event-driven network programs, it is necessary to check properties
which describe interactions between {\em multiple} packets.
There are several works which seek to perform
network updates in the context of multi-packet
properties \cite{ghorbani2014towards}
\cite{liu2015inter}.
There are also proposals for synthesizing SDN controller programs from
multi-packet examples \cite{yuan2015netegg} and from first-order
specifications \cite{padon2015decentralizing}. Lopes et al. 
presented techniques for verifying reachability in stateful network
programs \cite{lopes2015checking}, using a variant of Datalog. This
is a complimentary approach which could be used as a basis for
verifying reachability properties of our stateful programs.

\balance

\paragraph{Network Programming Languages.}
Network programs can often be constructed using high-level languages.
The Frenetic project
\cite{foster2011frenetic}
\cite{monsanto2012compiler}
\cite{foster2013languages}
allows higher-level specification of network policies.
Other related projects like Merlin \cite{soule2014merlin}
and NetKAT \cite{smolka2015} 
\cite{beckett2015temporal}
provide high-level languages/tools to compile such programs to network configurations.
Works such as
Maple
\cite{voellmy2013maple} and
FlowLog \cite{nelson2014tierless}
seek to address the {\em dynamic} aspect of network programming.

None of these systems and languages provide both (1) event-based constructs,
and (2) strong semantic guarantees about
consistency during updates, while our framework enables both.
Concurrently with this paper, an approach called SNAP \cite{arashloo2015snap}
was developed, which enables event-driven programming, and allows the
programmer to ensure consistency via an {\em atomic} language construct.
Their approach offers a more expressive language than our Stateful
NetKAT, but in our approach, we enable 
correct-by-construction event-based behavior and provide a dynamic
correctness property, showing (formally) that is strong enough for easy reasoning, 
yet flexible enough to enable efficient implementations.
We also prove the correctness of our implementation technique.

\paragraph{Routing.} The consistency/availability trade-off
is of interest in routing outside the SDN context as
well. In~\cite{JKKAV08}, a solution called consensus routing is presented, based on
a notion of causality between {\em triggers} (related to our
events). However, the solution is different in many aspects, e.g.
it allows a transient phase without safety guarantees.

\paragraph{High-Level Network Functionality.}
Some recent work has proposed building powerful high-level features into
the network itself, such as fabrics \cite{casado2012fabric}, intents
\cite{onos2014intent}, and other virtualization functionality
\cite{koponen2014network}. 
Pyretic \cite{monsanto2013} and projects built on top of it
such as PyResonance \cite{kim2013simpler}, SDX \cite{gupta2014sdx},
and Kinetic \cite{kim2015kinetic} provide high-level operations on
which network programs can be built. These projects do not guarantee
consistency during updates, and thus could be profitably combined
with an approach such as ours.  

\section{Discussion and Future Work}
\label{sec:discuss}

\paragraph{Generality of Our Approach.}  
\propose{
The event-driven
\xmnote{\FiveStar}{Q10}%
SDN update problem considered in this paper
is an instance of a more general distributed-systems programming problem,
namely {\em how to write correct and efficient programs for distributed
systems}.
We provide a PL approach (consistency property, programming language, and
compiler/runtime) which ensures that the programmer need not reason about
interleavings of events and updates for each application, and we
show that our consistency model and implementation technique work well
in the context of
SDN programs, but we do not believe they are limited to that %
specific arena. Our approach could also possibly be extended to other distributed
systems in which availability is prioritized, and consistency can be
relaxed in a well-defined way, as in our event-driven consistent
updates. Example domains include wireless sensor networks or other  
message-passing systems where the nodes have basic stateful functionality.
}

\paragraph{Future Work.}  
There are several directions for future work which could
address limitations of our current system.

\begin{compactenum}
\item We assume that the
set of (potential) hosts is known in advance, and use this information
to generate corresponding flow tables for each switch. This may
not be the right choice in settings where hosts join/leave. Our
approach could be extended to represent hosts {\em
  symbolically}.
\item We currently store all configurations on the
switches, so that they are immediately available during updates. Our
optimizations allow this to be done in a space-efficient way, but
there may be situations when it would be better for the controller to
reactively push new configurations to switches. This is an
interesting
problem due to interleavings of events and controller commands.
\item It would be interesting to consider formal
reasoning and automated verification for Stateful NetKAT.
\item \propose{We provide a solution
\xmnote{\FiveStar}{Q12}%
to the problem of performing
multiple updates, and the dynamic implementations we produce are
meant to ``run" in the network indefinitely.
However, there may be ways to update the running
dynamic program itself in some consistent way.
}
\end{compactenum}

\section{Conclusion}
\label{sec:conclusion}

This paper presents a full framework for correct event-driven
programming. Our approach provides a way of rigorously defining
correct event-driven behavior without the need for specifying logical
formulas.  We detail a programming language and compiler which allow
the user to write high-level network programs and produce correct and
efficient SDN implementations, and we demonstrate the benefits of our
approach using real-world examples.  This paper considers the
challenging problem of distributing an event-based stateful network program, and
solves it in a principled way.

\acks
Many thanks to the anonymous PLDI reviewers for offering
helpful and constructive comments, as well as Zach Tatlock for
shepherding our paper and providing useful feedback.
Our work is supported by the National Science Foundation under grants
CNS-1111698, CNS-1413972, CCF-1421752, CCF-1422046, CCF-1253165, and
CCF-1535952; the Office of Naval Research under grant
N00014-15-1-2177; and gifts from Cisco, Facebook, Fujitsu, Google, and
Intel.

\clearpage

\renewcommand*{\bibfont}{\small}
\bibliographystyle{abbrvnat}
\bibliography{paper}

\end{document}